\documentclass[journal,12pt, onecolumn]{IEEEtran}
\normalsize

\usepackage[english]{babel}
\usepackage[latin1]{inputenc}
\usepackage{amssymb,amstext,amsmath}
\usepackage[ruled,vlined]{algorithm2e}
\usepackage{verbatim}
\usepackage{lettrine}
\usepackage{fancyhdr}

\usepackage{float}

\usepackage{cite}

\ifCLASSINFOpdf
   \usepackage[pdftex]{graphicx}
   \DeclareGraphicsExtensions{.pdf,.jpeg,.png}
\else
   \usepackage[dvips]{graphicx}
   \DeclareGraphicsExtensions{.eps}
\fi

\newtheorem{theorem}{Theorem}

\newtheorem{corollary}{Corollary}

\begin{document}
\renewcommand{\baselinestretch}{2}

\title{Multiple Access Gaussian Channels with Arbitrary Inputs: Optimal Precoding and Power Allocation}

\author{Samah~A.~M.~Ghanem,~\IEEEmembership{Member,~IEEE,}
        
} 


\maketitle
\renewcommand{\baselinestretch}{2}
\begin{abstract}
In this paper, we derive new closed-form expressions for the gradient of the mutual information with respect to arbitrary parameters of the two-user multiple access channel (MAC). The derived relations generalize the fundamental relation between the derivative of the mutual information and the minimum mean squared error (MMSE) to multiuser setups. We prove that the derivative of the mutual information with respect to the signal to noise ratio (SNR) is equal to the MMSE plus a covariance induced due to the interference, quantified by a term with respect to the cross correlation of the multiuser input estimates, the channels and the precoding matrices. We also derive new relations for the gradient of the conditional and non-conditional mutual information with respect to the MMSE. Capitalizing on the new fundamental relations, we investigate the linear precoding and power allocation policies that maximize the mutual information for the two-user MAC Gaussian channels with arbitrary input distributions. We show that the optimal design of linear precoders may satisfy a fixed-point equation as a function of the channel and the input constellation under specific setups. We show also that the non-mutual interference in a multiuser setup introduces a term to the gradient of the mutual information which plays a fundamental role in the design of optimal transmission strategies, particularly the optimal precoding and power allocation, and explains the losses in the data rates. Therefore, we provide a novel interpretation of the interference with respect to the channel, power, and input estimates of the main user and the interferer.
\end{abstract}
\begin{IEEEkeywords}
Estimation Theory; Gradient of conditional mutual information; Gradient of non-conditional mutual information; Gradient of joint mutual information; Information Theory; Interference; MAC; MMSE; Mutiuser I-MMSE; Mutual Information; Power Allocation; Precoding.
\end{IEEEkeywords}


\section{Introduction}
{\linespread{2}\selectfont
\lettrine{L}{near} precoding for mutually interfering channels corrupted with independent Gaussian noise can be used to maximize the mutual information under a total power constraint \cite{91}. In the multiple user setting, the mutual information is maximized by imposing a covariance structure on the Gaussian input vector that satisfies two conditions: First, transmitting along the right eigenvectors of the channel. Second, the power allocation follows the waterfilling based on the SNR in each parallel non-interacting sub-channel \cite{42}, \cite{3}. The MAC channel stands as a special case of interference channels. With arbitrary inputs, the mutual information of the MAC channel is minimized due to the fact that at a certain point the inputs may lie in the null space of the channel, and this causes the mutual information to decay at such point. Therefore, the capacity-achieving strategies can be employed either at the transmitter side and/or the receiver side: by orthogonalizing the inputs using TDM/FDM techniques, or via using interference mitigation techniques, such as, dealing with interference as noise, rate splitting, successive interference cancellation, decoding interference, as well as receive diversity techniques \cite{99},~\cite{1421925},~\cite{485709},~\cite{Goldsmith},~\cite{4276942},~\cite{5199541}.

Precoding is another technique that can be used at the transmitter side to maximize the mutual information by exploiting the transmit diversity. On the one hand, precoding plays the role of power allocation diagonalizing the channel while, on the other hand, rotating the channel eigenvectors; thus shaping the constellation for better detection. Since practical considerations enforce the use of discrete constellations, such as PSK and QAM, which depart significantly from the optimal Gaussian distributions, it is of great importance to revisit the linear precoding problem under the constraint that each input follows a specified discrete constellation. In \cite{3}, \cite{35}, and \cite{92}, the optimal power allocation for parallel non-interacting channels with arbitrary inputs, and the optimal precoding for multiple inputs multiple outputs (MIMO) channels with different setups were obtained exploiting the relation between the mutual information and the MMSE \cite{35}, \cite{34}. The results in \cite{91} demonstrate that capacity-achieving strategies for Gaussian inputs may be suboptimal for discrete inputs. In \cite{3} a general solution for the optimum precoder has been derived for the high and low-snr regimes capitalizing on first order expansions of the MMSE and relating it to mutual information. The mercury/waterfilling solution for power allocation maximizes the mutual information compensating for the non-Gaussianess for the discrete constellations \cite{39}. However, linear precoding techniques introduce rescaling among the channel inputs that achieves cleaner detection of the received signals and thus higher information rates. Therefore, this paper addresses the optimal precoding and optimal power allocation problem for multiuser channels driven by arbitrary inputs, extending previous works presented in \cite{98}, and \cite{107},~\cite{samahthesis}.

Due to the lack of explicit closed form expressions of the mutual information for binary inputs, which were provided only for the BPSK for the single input signle output (SISO) case,~\cite{34}, \cite{samahMAPtele}, \cite{samahthesis}, it is of particular importance to address connections between information theory and estimation theory to understand the communication framework under such inputs and to design optimal precoding and power allocation strategies that aim to maximize the reliable infromation rates.

In this paper\footnote{Preliminary results of this work has been presented in WCSP \cite{98}.}, we first revisit the connections between the mutual information and the MMSE for the multiuser setup, to derive relations that account for the existence of an interfer channel. Therefore, the known fundamental relation between the derivative of the mutual information and the MMSE defined for point to point channels with any noise or input distributions in \cite{34}, is generalized to the multiuser case. Moreover, the relations for linear vector Gaussian channels in \cite{35} is generalized for the multiuser channels where we extend these relations to the per-user gradient of the mutual information with respect to the MMSE, channels and precoders (power allocation) matrices of the user and the interferers. We then propose a linear precoder structure that aims to maximize the input output mutual information for the two-user MAC setting with channel state information known at both receiver and transmitter sides.

Throughout the paper, the following notation is employed, boldface uppercase letters denote matrices, lowercase letters denote scalars. The superscript, $(.)^{-1}$, $(.)^{T}$, $(.)^{\textit{T}}$, and $(.)^{\dag}$ denote inverse, transpose, conjugate and conjugate transpose operations. $(.)^{\star}$ denotes optimum. The $(\nabla)$ denotes the gradient of a scalar function with respect to a variable.  The $\mathbb{E}[.]$ denotes the expectation operator. The $||.||$ and $Tr\left\{.\right\}$ denote the Euclidean norm, and the trace of a matrix, respectively.

The paper is organized as follows; section II introduces the system model and the problem formulation and, section III introduces the new fundamental relations between the mutual information and the MMSE. Section IV introduces the optimal precoding with arbitrary inputs and under different user setups. Section V introduces analysis at the asymptotic regime of low SNR. Section VI introduces the power allocation with arbitrary inputs and specializes it to the mercury/waterfilling under specfic user setups. Section VII introduces simulation and analytical results where we introduce an interpretaion of the interference, similar to known interpretations of the optimal power allocation over user main channels.\par}

\section{System Model and Problem Formulation}
{\linespread{2}\selectfont
Consider the deterministic complex-valued vector channel,
\begin{equation}
\label{1}
{\bf{y}}=\sqrt{snr}~{\bf{H_1P_1}}{\bf{x_1}}+\sqrt{snr}~{\bf{H_2P_2}}\bf{x_2}+\bf{n},
\end{equation}

where the $n_r~\times~1$ dimensional vector $\bf{y}$ and the $n_t~\times~1$ dimensional vectors $\bf{x_1}$, $\bf{x_2}$ represent, respectively, the received vector and the independent zero-mean unit-variance transmitted information vectors from each user input to the MAC channel. The distributions of both inputs are not fixed, not necessarily Gaussian nor identical. The $n_r \times n_t$ complex-valued  matrices $\bf{H_1}$, $\bf{H_2}$ correspond to the deterministic channel gains for both input channels (known to both encoder and decoder) and $\bf{n} \sim \mathcal{CN}(0,I)$ is the $n_r \times 1$ dimensional complex Gaussian noise with independent zero-mean unit-variance components. The optimization of the mutual information is carried out over all $n_t \times n_t$ precoding matrices $\bf{P_1}$, $\bf{P_2}$ that do not increase the transmitted power. The precoding problem can be cast as a constrained non-linear optimization problem as follows:
\begin{equation}
\label{2}
max~~I(\bf{x_1, x_2; y})
\end{equation}
Subject to: 
\begin{equation}
\label{3}
Tr \left\{\mathbb{E}_x[\bf{P_1x_1}x_1^{\dag}\bf{P_1}^{\dag}]\right\}=Tr\left \{\bf{P_1P_1}^{\dag}\right\} \leq Q_1
\end{equation}
\begin{equation}
\label{4}
Tr\left\{\mathbb{E}_x[\bf{P_2}x_2x_2^{\dag}\bf{P_2}^{\dag}]\right\}=Tr\left\{\bf{P_2P_2}^{\dag}\right\} \leq Q_2
\end{equation}

To address the precoding and power allocation problems, we will revisit the fundamental relation between the mutual information and the MMSE with respect to the SNR for the multiuser case. The following section provides a set of new results that generalize of the known relations derived first by Guo, Shamai, and Verdu \cite{34}, and extends the results of linear vector channels by Palomar and Verdu \cite{35} to the general case of multiuser channels.\par}

\section{New Fundamental Relations between the Mutual Information and the MMSE}
The first contribution is given in the following theorem, which provides a generlization of the I-MMSE identity to the multiuser case and traverses back to the same identilty of the single user case.
\begin{theorem}
\label{theorem1.0}
The relation between the derivative of the mutual information with respect to the snr and the non-linear MMSE for a multiuser Gaussian channel satisfies:
\begin{equation}
\frac{dI(snr)}{dsnr}=mmse(snr)+\psi(snr)
\end{equation}
\end{theorem}
Where,
\begin{equation}
\label{mmsee}
mmse(snr)=Tr\left\{\bf{H_1P_1}\bf{E_1}\bf{(H_1P_1)^{\dag}}\right\}+Tr\left\{\bf{H_2P_2}\bf{E_2}\bf{(H_2P_2)^{\dag}} \right\},
\end{equation}
\begin{multline}
\label{psii}
\psi(snr)=Tr\left\{\bf{H_1P_1}\mathbb{E}_y[\bf{\mathbb{E}_{x_1|y}[\bf{x_1|y}]\mathbb{E}_{x_2|y}[\bf{x_2|y}]^{\dag}}]\bf{(H_2P_2)^{\dag}}\right\} \nonumber \\
-Tr\left\{\bf{H_2P_2}\mathbb{E}_y[\bf{\mathbb{E}_{x_2|y}[\bf{x_2|y}]\mathbb{E}_{x_1|y}[\bf{x_1|y}]^{\dag}}]\bf{(H_1P_1)^{\dag}}\right\},
\end{multline}

\begin{proof}
The proof has been done for the two-user MAC channel as an example of a multiuser channel. See Appendix A.
\end{proof}

The per-user MMSE is given respectively as follows:
\begin{equation}
\label{7}
\bf{E_1}=\bf{\mathbb{E}_y[(x_1-\widehat{x}_1)(x_1-\widehat{x}_1)^{\dag}]}
\end{equation}
\begin{equation}
\label{8}
\bf{E_2}=\bf{\mathbb{E}_y[(x_2-\widehat{x}_2)(x_2-\widehat{x}_2)^{\dag}]}.
\end{equation}
The input estimates of each user input is given respectively as follows:
\begin{equation}
\label{eq.37}
\widehat{{\bf{x}}}_1={\bf{\mathbb{E}_{x_{1}|y}}[{\bf{x_{1}|y}}]}=\sum_{{\bf{x_{1},x_{2}}}} \frac{{\bf{x_{1}}}p_{y|x_{1},x_{2}}({\bf{y|x_{1},x_{2}}})p_{x_1}({\bf{x_{1}}})p_{x_2}({\bf{x_{2}}})}{p_{y}({\bf{y}})}
\end{equation}
\begin{equation}
\label{eq.38}
\widehat{{\bf{x}}}_2={\bf{\mathbb{E}_{x_{2}|y}}[{\bf{x_{2}|y}}]}=\sum_{{\bf{x_{1},x_{2}}}} \frac{{\bf{x_{2}}}p_{y|x_{1},x_{2}}({\bf{y|x_{1},x_{2}}})p_{x_1}({\bf{x_{1}}})p_{x_2}({\bf{x_{2}}})}{p_{y}({\bf{y}})}.
\end{equation}
The conditional probability distribution of the Gaussian noise is defined as:
\begin{equation}
p_{y|x_1,x_2}({\bf{y|x_1, x_2}})=\frac{1}{\pi^{n_r}}e^{-\left\|{\bf{y}}-\sqrt{snr}{\bf{H_1P_1x_1}}-\sqrt{snr}{\bf{H_2P_2x_2}}\right\|^{2}}
\end{equation}
The probability density function for the received vector $\bf{y}$ is defined as:
\begin{equation}
{p_{y}({\bf{y}})}=\sum_{{\bf{x_{1},x_{2}}}} {p_{y|x_{1},x_{2}}}({\bf{y|x_{1},x_{2}}})p_{x_1}({\bf{x_{1}}})p_{x_2}({\bf{x_{2}}}).
\end{equation}

Henceforth, the system MMSE with respect to the SNR is given by:
\begin{equation}
\label{14}
mmse(snr)=\bf{\mathbb{E}_y}\left[\left\|\bf{H_1P_1}(\bf{x_1}-\bf{\mathbb{E}_{x_1|y}[x_1|y]})\right\|^{2}\right]+\bf{\mathbb{E}_y}\left[\left\|\bf{H_2P_2}(\bf{x_2-\mathbb{E}_{x_2|y}[x_2|y]})\right\|^{2}\right],
\end{equation}
\normalsize
\begin{equation}
~~=Tr\left\{\bf{H_1P_1}\bf{E_1}{(\bf{H_1P_1})}^{\dag}\right\} + Tr\left\{{\bf{H_2P_2}\bf{E_2}{(\bf{H_2P_2})}^{\dag}}\right\}  
\end{equation}
\normalsize
as given in Theorem~\ref{theorem1.0}.

Note that the term $mmse(snr)$ is due to the users MMSEs, particularly, $mmse(snr)=mmse_1(snr)+mmse_2(snr)$ and $\psi(snr)$ are covariance terms that appear due to the covariance of the interferers. Those terms are with respect to the channels, precoders, and non-linear estimates of the user inputs. When the covariance terms vanish to zero, the mutual information with respect to the SNR will be equal to the mmse with respect to the SNR, this applies to the single user and point to point communications. Therefore, the result of Theorem~\ref{theorem1.0} is a generalization of previous result and boils down to the result of Guo et al, \cite{34} under certain conditions which are: (i) when the cross correlation between the inputs estimates equals zero (ii) when interference can be neglected (i.e. interference is very weak, very strong, or aligned) and (iii) under the signle user setup.

Such generalized fundamental relation between the change in the mutliuser mutual information and the SNR is of particular relevance. Firstly, such result allows us to understand the behavior of per-user rates with respect to the interference due to the mutual interference and the interference due to other users behaviour in terms of power levels and channel strengths. In addition, the result allows us to be able to quantify the losses incured due to the interference in terms of bits. Therefore, when the term $\psi(snr)$ equals zero. The derivative of the mutual information with respect to the SNR equals the total $mmse(snr)$:
\begin{equation}
\frac{dI(snr)}{d snr}=mmse(snr),
\end{equation}
which matches the result by Guo et. al in~\cite{34}.

\subsection{Remark on Interference Channels}
The implication of the derived relation on the interference channel is of particular relevance. It is worth to note that we can capitalize on the new fundamental relation to extend the derivative with respect to the SNR to the conditional mutual information. To make this more clear, we capitalize on the chain rule of the mutual information which states the following:
\begin{equation}
\label{chainR}
I({\bf{x_1, x_2; y}})=I({\bf{x_1; y}})+I({\bf{x_2; y|x_1}}) 
\end{equation}
Therefore, through this observation we can conclude the following theorem.
\begin{theorem}
\label{theoremConditional}
The relation between the derivative of the conditional mutual information and the minimum mean squared error satisfies:
\begin{equation}
\frac{dI(\bf{x_2; y}|x_1)}{dsnr}=mmse_2(snr)+\psi(snr)
\end{equation}
\end{theorem}

\begin{proof}
Taking the derivative of both sides of \eqref{chainR}, and subtracting the derivative of $I(\bf{x_1; y})$ which is equal to $mmse_1(\gamma snr)$, $\gamma$ is a scaling factor, due to the fact that $x_1$ is decoded first considering the other users' input $x_2$ as noise. Therefore, Theorem~\ref{theoremConditional} has been proved.
\end{proof}
Note that the derivative of the conditional mutual information as well as the non-conditional mutual information can be scaled due to different SNRs in a two-user interference channel. However, the scaling is straightforward to apply. For further details, refer to \cite{samahthesis}, \cite{SamahLaura}, and \cite{BustinShamai}.

To solve the optimization problem in \eqref{2} subject to \eqref{3} and \eqref{4}, we need to know the relation between the mutual information and the non-linear MMSE in the context of MAC channels. The following theorem puts forth a new fundamental relation between the mutual information and the non-linear MMSE for MIMO MAC channels~\footnote{Noice that through the flow of the paper, we use MAC channel instead of MIMO MAC since such relations are general and apply even when users are transmitting over scalar channels.}. In particular, the gradient of the mutual information with respect to arbitrary parameters has been derived. First we will present the gradient of the mutual information, with respect to the channel in Theorem~\ref{theorem1.4}, then with respect to the precoder (or power allocation) in Theorem~\ref{theorem2.4}.

The following theorems in addition to Theorem \ref{theorem1.0} generalizes the connections between information theory and estimation theory to the multiuser case.
\begin{theorem}
\label{theorem1.4}
The relation between the gradient of the mutual information with respect to the channel and the non-linear MMSE for a two user-MAC channel with arbitrary inputs~\eqref{1} satisfies:
\begin{equation}
\label{eq.5}
{{\bf{\nabla_{H_1}}} I({\bf{x_{1},x_{2};y}})=\bf{H_1P_1E_1P_1}^{\dag}-\bf{H_2 P_2}\bf{\mathbb{E}[\widehat{x}_2\widehat{x}_1^{\dag}]}\bf{P_1}^{\dag}}
\end{equation}
\begin{equation}
\label{eq.6}
{{\bf{\nabla_{H_2}}} I({\bf{x_{1},x_{2};y}})=\bf{H_2P_2E_2P_2^{\dag}}-\bf{H_1 P_1} \bf{\mathbb{E}[\widehat{x}_1\widehat{x}_2^{\dag}]}\bf{P_2}^{\dag}}
\end{equation}
\end{theorem}
\begin{proof}
See Appendix B.
\end{proof}

\begin{theorem}
\label{theorem2.4}
The relation between the gradient of the mutual information with respect to the precoding matrix and the non-linear MMSE for a two user-MAC channel with arbitrary inputs~\eqref{1} satisfies:
\begin{equation}
\label{eq.5.1}
{\bf{\nabla_{P_1}}}I({\bf{{x_{1},x_{2};y}}})=\bf{H_1^{\dag}}{\bf{H}_{1}}{\bf{P}_1}{\bf{E}_{1}}-\bf{H_1^{\dag}}{\bf{H}_{2}}{\bf{P}_2}\bf{\mathbb{E}[\widehat{x}_2\widehat{x}_1^{\dag}]}
\end{equation}
\begin{equation}
\label{eq.6.1}
{\bf{\nabla_{P_2}}}I({\bf{x_{1},x_{2};y}})=\bf{H_2^{\dag}}{\bf{H}_{2}}{\bf{P}_2}{\bf{E}_{2}}-\bf{H_2^{\dag}}{\bf{H}_{1}}{\bf{P}_1}\bf{\mathbb{E}[\widehat{x}_1\widehat{x}_2^{\dag}]}
\end{equation}
\end{theorem}

\begin{proof}
See Appendix C.
\end{proof}

Theorems~\ref{theorem1.4} and \ref{theorem2.4} provides intuitions about the change of the mutual information with respect to the changes in the channel or the precoding (power allocation). There is a strightforward connection between both changes when each gradient is scaled with respect to the changing arbitrary parameter, we can write this connection as follows:
\begin{equation}
 {\bf{\nabla_{P_1}}}I({\bf{{x_{1},x_{2};y}}}){\bf{P_1^{\dag}}}={\bf{H_1^{\dag}}}{\bf{\nabla_{H_1}}}I({\bf{x_{1},x_{2};y}}),
\end{equation}
\begin{equation}
 {\bf{\nabla_{P_2}}}I({\bf{x_{1},x_{2};y}}){\bf{P_2^{\dag}}}={\bf{H_2^{\dag}}}{\bf{\nabla_{H_2}}}I({\bf{x_{1},x_{2};y}}),
\end{equation}

Note that we can derive the gradient of the mutual information with respect to any arbitrary parameter following similar steps of the proof of the previous two theorems. Note also that the derived relations in \eqref{eq.5},~\eqref{eq.6},~\eqref{eq.5.1}, and \eqref{eq.6.1} reduce to the relation between the gradient of the mutual information and the non-linear MMSE derived for the linear vector Gaussian channels \cite{35} if the cross correlation between the input estimates is zero. This is possible when the users inputs are orthogonal or interference alignment is performed. In addition, the derived relation between the gradient of the mutual information and the MMSE of the two-user MAC setting in Theorem~\ref{theorem1.4}~will lead to a new formulation of the input non-linear estimates, that is discussed in the following Theorem.

\begin{theorem}
\label{theorem3.4}
The estimates of the inputs $\bf{x_1}$ and $\bf{x_2}$ of the two-user MAC channel with arbitrary inputs given the output $\bf{y}$ can be expressed as:
\begin{equation}
\label{eq.13}
{\bf{H_1P_1}}\mathbb{E}_{x_1|y}{\bf{[x_1|y]}}+{\bf{H_2P_2}}\mathbb{E}_{x_2|y}{\bf{[x_2|y]}}= {\bf{y}} + \frac{\nabla_{y}p_{y}({\bf{y}})}{p_{y}(\bf{y})}
\end{equation}
\end{theorem}

\begin{proof}
See Appendix D.
\end{proof}

\eqref{eq.13} can be solved numerically to obtain the estimates of the inputs. However, we can extract a special case of Theorem~\ref{theorem3.4} in the following corrolary.

\begin{corollary}
If the channel is a linear vector Gaussian channel, Theorem~\ref{theorem3.4} leads to the following formulation:
\begin{equation}
\label{eq.18}
\mathbb{E}_{x|y}{\bf{[x|y]}}= \left({\bf{y}} + \frac{\nabla_{y}p_{y}({\bf{y}})}{p_{y}({\bf{y}})}\right) (\bf{HP})^{-1}
\end{equation}
\end{corollary}

\eqref{eq.18} can be manipulated in terms of its zero forcing linear counter measure, where the distance between a linear and a non-linear estimate is approximately equal to $\epsilon$, where $\epsilon=(\frac{\nabla_{y}p_{y}(\bf{y})}{p_{y}(\bf{y})})\bf{(HP)}^{-1}$. Further, the non-linear estimates given in \eqref{eq.37} and \eqref{eq.38} give a statistical intuition to the problem. However, the most used in practical setups are linear estimation sub-optimal methods, where linear estimates can be found via the linear MMSE. In particular, the input estimates can be found by deriving the optimal Wiener receive filters solving a minimization optimization problem of the mean squared error (MSE). The following theorem provides the linear MMSE estimates.

\begin{theorem}
\label{theorem4.00}
The linear estimates of the inputs $\bf{x_1}$, $\bf{x_2}$ of the two-user MAC channel given the output $\bf{y}$ can be expressed as: 
\begin{equation}
\bf{\widehat{x}_1}=\bf{P_1^{\dag}H_1^{\dag}(I+P_1^{\dag}H_1^{\dag}P_1H_1+P_2^{\dag}H_2^{\dag}P_2H_2)^{-1}y}
\end{equation}
\begin{equation}
\bf{\widehat{x}_2}=\bf{P_2^{\dag}H_2^{\dag}(I+P_1^{\dag}H_1^{\dag}P_1H_1+P_2^{\dag}H_2^{\dag}P_2H_2)^{-1}y}
\end{equation}
\end{theorem}

\begin{proof}
See Appendix E.
\end{proof}

Worth to note that the linear estimation process of the inputs corresponds to multiplying the user Wiener MMSE receive filters to the received vector. Note that each receive filter contains the precoding (or power allocation), the user channel, in addition to a structure of the linear MMSE, that is, the inverse of the covariance structure of the received vector.

\subsection{Gradient of the Conditional and Non-Conditional Mutual Information}
Theorem~\ref{theorem2.4} shows how much rate is lost due to the other user, this is due to the fact that some terms in the gradient of the mutual information preclude the effect of the mutual interference of the main links. Therefore, we can account for such quantified rate loss via optimal power allocation and optimal precoding. Then, in order to be able to understand the achieved rates of each user in a MAC channel, we capitalize on the chain rule of the mutual information to derive the conditional mutual information as follows,
\begin{equation}
\label{a390}
I({\bf{x_{1},x_{2};y}})=I({\bf{x_{2};y}})+I({\bf{x_{1};y|x_{2}}})
\end{equation}
\normalsize
and,
\begin{equation}
\label{a391}
I({\bf{x_{1},x_{2};y}})=I({\bf{x_{1};y}})+I({\bf{x_{2};y|x_{1}}})
\end{equation}
\normalsize
Clearly, we know that $I({\bf{x_{2};y}})$ is the mutual information when user 2 is decoded first considering user 1 signal as noise. Therefore, we can write it as follows,
\begin{equation}
\label{a41}
I({\bf{x_{2};y}})=\left[\log \frac{p_{y|x_{2}}({\bf{y|x_{2}}})}{p_{y}({\bf{y}})}\right]
\end{equation}
\begin{equation}
\label{a42}
p_{y|x_{2}}({\bf{y|x_{2}}})=\frac{1}{\pi^{n_r}}e^{-{({\bf{y}}-\sqrt{snr}\bf{{H}_{2}}{{P_2}}\bf{x_{2}})^{\dag}(\bf{{P}_{1}^{\dag}{H}_{1}^{\dag}{H}_{1}{P_1}+I})^{-1}({\bf{y}}-\sqrt{snr}\bf{{H}_{2}}{{P_2}}\bf{x_{2}})}}
\end{equation}
\begin{equation}
\label{a43}
p_{y}({\bf{y}})=\displaystyle\sum\limits_{x_{2}} p_{y|x_{2}}({\bf{y|x_{2}}})p_{x_{2}}({\bf{x_{2}}})
\end{equation}
Based on such definition, we conclude the following theorem which provides a new fundamental relation between the gradient of the mutual information with respect to the precoder of the other user given that the other user will be secondly decoded.
\begin{theorem}
\label{gradientwithnoise}
The gradient of the mutual information with respect to the precoder, for a scaled user power when the other user input is considered as noise is as follows,
\begin{equation}
 {{\nabla_{{\bf{P_1}}}}}I({\bf{x_{2};y}})=\bf{{H}_{2} P_2 E_{2}}\bf{{P}_{2}^{\dag}{H}_{2}^{\dag}} \bf{{H}_{1}}^{\dag}\bf{H_{1}P_{1}}(\bf{{P}_{1}^{\dag}{H}_{1}^{\dag}{H}_{1}{P_1}+I})^{-1}
\end{equation}
\begin{equation}
 {\nabla_{{\bf{P_2}}}}I({\bf{x_{1};y}})=\bf{{H}_{1} P_1 E_{1}} \bf{{P}_{1}^{\dag}{H}_{1}^{\dag}} \bf{{H}_{2}}^{\dag}\bf{H_{2}P_{2}}(\bf{{P}_{2}^{\dag}{H}_{2}^{\dag}{H}_{2}{P_2}+I})^{-1}
\end{equation}
\end{theorem}
\begin{proof}
The proof follows similar steps of the proof of Theorem~\ref{theorem2.4}. 
\end{proof}
When each user is transmitting over a single channel, the new relation in Theorem \ref{gradientwithnoise} will be more clearly understood in terms of the effect of the interference plus noise power scaling on the gradient, see \cite{SamahLaura}.
\begin{corollary}
The gradient of the conditional mutual information will be as follows,
\begin{equation}
\label{a260}
 {\bf{\nabla_{P_1}}}I({\bf{x_{1};y}|x_{2}})= {\bf{\nabla_{P_1}}}I({\bf{x_{1},x_{2};y}})- {\bf{\nabla_{P_1}}}I({\bf{x_{2};y}})
\end{equation}
\begin{equation}
\label{a271}
 {\bf{\nabla_{P_2}}}I({\bf{x_{2};y}|x_{1}})={\bf{\nabla_{P_2}}}I({\bf{x_{1},x_{2};y}})-{\bf{\nabla_{P_2}}}I({\bf{x_{1};y}})
\end{equation}
\normalsize
\end{corollary}
\begin{proof}
The proof of the corllary follows from the chain rule of mutual information and the gradient of the mutual information derived for the sum rate and non-conditional rates.
\end{proof}

\section{Optimal Precoding with Arbitrary Inputs}
Capitalizing on the relation between the gradient of the mutual information with respect to the precoding matrix and the MMSE, we can solve the optimization problem in \eqref{2} subject to \eqref{3} and \eqref{4}. The following theorem provides the optimal precoding of the two-user MAC channel with arbitrary inputs which can be also specialized to Gaussian inputs, see~\cite{samahthesis},\cite{SamahLaura},\cite{107}.

\begin{theorem}
\label{theorem5.4}
Let the distribution of the arbitrary inputs $\bf{x_1}$, $\bf{x_2}$ to the channel in \eqref{1} be $\bf{p_{x_1}(x_1)}$, $\bf{p_{x_2}(x_2)}$, respectively. The optimal precoding matrices satisfy the following:
\begin{equation}
\label{9}
{{\bf{P}}_1}^{\star}=\nu^{-1}{\bf{H}}_{1}^{\dag}{\bf{H}_{1}}{\bf{P}_1}^{\star}{\bf{E}_{1}}-\nu^{-1}{\bf{H}}_{1}^{\dag}{\bf{H}_{2}}{\bf{P}_2}^{\star}\bf{\mathbb{E}[\widehat{x}_2\widehat{x}_1^{\dag}]}
\end{equation}
%
\begin{equation}
\label{10}
{{\bf{P}}_2}^{\star}=\nu^{-1}{\bf{H}}_{2}^{\dag}{\bf{H}_{2}}{\bf{P}_2}^{\star}{\bf{E}_{2}}-\nu^{-1}{\bf{H}}_{2}^{\dag}{\bf{H}_{1}}{\bf{P}_1}^{\star}\bf{\mathbb{E}[\widehat{x}_1\widehat{x}_2^{\dag}]}
\end{equation}
\end{theorem}

Where $\nu$ are the received $snr$ normalized by the Lagrange multipliers of the optimization problem in \eqref{2}. 

\begin{proof}
The proof of Theorem~\ref{theorem5.4}~relies on the KKT conditions \cite{52} and the relation between the gradient of the mutual information with respect to the precoding matrices and the MMSE, see Appendix F.
\end{proof}

There are unique $\bf{P_1}^{\star}$, $\bf{P_2}^{\star}$ that satisfy the KKT conditions when the problem is strictly concave, corresponding to the global maximum. Using Monte-Carlo method, we can compute the MMSE matrices \eqref{7} and \eqref{8} for discrete constellations. To obtain the solution of \eqref{9} and \eqref{10}, we can use an iterative approach, similar to \cite{91}, \cite{3}, and \cite{92}. It is very important to notice that the optimal precoders of \eqref{9} and \eqref{10} satisfies a fixed point equation under two specific setups: 

\textit{First}, for the multiple-channels-per user in the two-user-MAC if and only if the second term in the gradient of the mutual information with respect to the precoding matrix is zero, this occurs when both inputs are orthogonal, or in this case, we can think about it as if the second user multiple-inputs have no influence to the first user multiple inputs, therefore, each user precoding is over his mutually interfering channels only. If each user is using OFDM signaling, i.e., transmitting over parallel independent Gaussian channels, the problem breaks into two separate problems, and the precoding solution for each is a fixed point equation. 

\textit{Second}, for the two-user-MAC case with each user is mutually interfering with the other, in this case the relation between the gradient of the mutual information with respect to the precoding matrix can be manipulated as in \cite{35}, thus, the precoder structure is similar to the one introduced in \cite{3}, \cite{107}. We define in the following theorem a general form of the optimal precoder's structure for the two-user MAC channel in \eqref{9} and \eqref{10} under the two specific setups where the precoder would satisfy a fixed point equation, therefore, the solution applies to the low-snr and high-snr regimes. However, it also reduces to the optimal power allocation solution when the precoding is precluded as will be shown later. Therefore, the precoder should admit a structure that performs matching of the strongest source modes to the weakest noise modes, and this alignment enforces permutation process to appear in the diagonal power allocation setup.

\begin{theorem}
\label{theorem6.4}
The non-unique first-order optimal precoders that maximize the mutual information for a two-user MAC subject to per-user power constraint, and under the assumption that the second term in \eqref{9}, and \eqref{10} equals zero, can be given as follows:
\begin{equation}
\label{11}
~~~~~\bf{P_1}^{\star}=\bf{U_{1}}\bf{D_1}\bf{R_{1}^{\dag}}
\end{equation}
\begin{equation}
\label{12}
~~~~~\bf{P_2}^{\star}=\bf{U_{2}}\bf{D_2}\bf{R_{2}^{\dag}}  
\end{equation}
With $\bf{U_{1}}$, $\bf{U_{2}}$ are unitary matrices, $\bf{D_1}$, $\bf{D_2}$ are diagonal matrices, and $\bf{R_1}$, $\bf{R_2}$ are rotation matrices.
\end{theorem}
\begin{proof}
The theorem follows the relation between the gradient of the mutual information with respect to the precoding matrix and the decomposition of the components of the channel, precoder, and MMSE matrices, and is a direct consequence to [Theorem 1, 5], see Appendix G. 
\end{proof}
The optimal precoders are as follows, $\bf{P_1}=\bf{U_1}\bf{D_1}\bf{R_{1}^{\dag}}$ and $\bf{P_2}=\bf{U_2}\bf{D_2}\bf{R_{2}^{\dag}}$ with $\bf{U_1=V_{H_1}}$, $\bf{U_2=V_{H_2}}$, correspond to the per user channel right singular vectors, respectively.\\ $\bf{D_1}=diag(\sqrt{p_{1,1}}$,$~...~$,${\sqrt{p_{1,n_t}}})$ and $\bf{D_2}=diag(\sqrt{p_{2,1}}$,$~...~$,${\sqrt{p_{2,n_t}}})$ are power allocation matrices, such that $\bf{p}_{l,m}$ is admitted by user $l$ over sub/channel $m$; this in fact corresponds to the mercury/waterfilling. $\bf{R_1}=\bf{\Pi U_{E_1}}$ and $\bf{R_2}=\bf{\Pi U_{E_2}}$ are rotation matrices that contain in their structure the eigenvectors of the user MMSE matrix which can be permuted and/or projected with $\Pi$ based on the correlation of the inputs and their estimates. Therefore, the rotation matrix insures firstly, allocation of power into the strongest channel singular vectors, and secondly, diagonalizes the MMSE matrix to insure un-correlating the error or in other words independence between inputs, enforcing the setup defined in \eqref{9} and \eqref{10} for the two special setups discussed previously about each user inputs, i.e., when the interference from each user to the other is cancelled or precluded, and approximated by zero.

\section{The Low-SNR Regime}
We now consider the optimal precoding policy for the two-user MAC Gaussian channel with arbitrary input distributions in the regime of low-snr. Consider a zero-mean uncorrelated complex inputs, with $\bf{\mathbb{E}[x_1x_1^{\dag}]=\mathbb{E}[x_2x_2^{\dag}]=I}$,~and~$\bf{\mathbb{E}[x_1x_1^{T}]=\mathbb{E}[x_2x_2^{T}]}$
$\bf{=0}$. We consider the low-snr expansion to the MMSE of equation \eqref{14}. Note that it can be easily deduced that the Taylor expansion of the non$-$linear MMSE in \eqref{14} will lead to the first order Taylor expansion of the linear MMSE for the Gaussian inputs setup. Thus, the low-snr expansion of the MMSE matrix can be expressed as:
\begin{equation}
\label{eq.15}
{\bf{E}}={\bf{I}}-{({\bf{H_1P_1}})}^{\dag}{\bf{H_1P_1}}.{snr}-{\bf{(H_2P_2)^{\dag}H_2P_2}}.{snr}+ \mathcal{O}(snr^2),
\end{equation}

with $\bf{E=E_1+E_2}$. Consequently,
\begin{equation}
mmse(snr)=Tr\left\{\bf{H_1P_1}\bf{E_1}{(\bf{H_1P_1})}^{\dag}\right\} +Tr \left\{\bf{H_2P_2}\bf{E_2}{(\bf{H_2P_2})}^{\dag}\right\} 
\end{equation}
\begin{multline}
~~=Tr \left\{\bf{H_1P_1}{(\bf{H_1P_1})}^{\dag}\right\} +Tr \left\{\bf{H_2P_2}{(\bf{H_2P_2})}^{\dag}\right\} \\
- Tr \left\{(\bf{H_1P_1}{(\bf{H_1P_1})}^{\dag})^{2}\right\}.snr - Tr\left\{(\bf{H_2P_2}{(\bf{H_2P_2})}^{\dag})^{2}\right\}.snr + \mathcal{O}(snr^2).
\end{multline}
Note that due to our new result of Theorem~\ref{theorem1.0}, we cannot apply immediately the fundamental relationship between mutual information and MMSE in \cite{35}, \cite{34}. Therefore, applying our new result, the low-snr Taylor expansion of the mutual information is given by:
\begin{multline}
\label{16}
I(snr)=Tr \left\{\bf{H_1P_1}{(\bf{H_1P_1})}^{\dag}\right\}.snr +Tr \left\{\bf{H_2P_2}{(\bf{H_2P_2})}^{\dag}\right\}.snr \\
- Tr\left\{(\bf{H_1P_1}{(\bf{H_1P_1})}^{\dag})^{2}\right\}.{snr^2} - Tr\left\{(\bf{H_2P_2}{(\bf{H_2P_2})}^{\dag})^{2}\right\}.{snr^2}+ \\
+Tr\left\{\bf{H_1P_1}{(\bf{H_1P_1})}^{\dag}\bf{H_2P_2}{(\bf{H_2P_2})}^{\dag}\right\}.{snr^2} \\
-Tr\left\{\bf{H_2P_2}{(\bf{H_2P_2})}^{\dag}\bf{H_1P_1}{(\bf{H_1P_1})}^{\dag}\right\}.{snr^2}+\mathcal{O}(snr^3)
\end{multline}
\normalsize
The wideband slope $-$ which indicates how fast the capacity is achieved in terms of required bandwidth $-$ is inversely proportional to the second order terms of the mutual information in the low-snr Taylor expansion \eqref{16}. Therefore, this term is a key low-power performance measure since the bandwidth required to sustain a given rate with a given low power, i.e., minimal energy per bit, is inversely proportional to this term \cite{93}. 

According to \eqref{16}, for first-order optimality, subject to \eqref{3} and \eqref{4}, the form of the optimal precoder follows from the low-snr expansions the form of optimal precoder for the complex Gaussian inputs settings. To prove this claim, lets re-define our optimization problem as follows:
\begin{equation}
\label{17}
max~~~Tr\left\{\bf{H_1P_1}{(\bf{H_1P_1})}^{\dag}\right\}.snr + Tr\left\{\bf{H_2P_2}{(\bf{H_2P_2})}^{\dag}\right\}.snr
\end{equation}
Subject to:
\begin{equation}
\label{18}
Tr\left\{\bf{P_1P_1}^{\dag}\right\} \leq 1
\end{equation}
\begin{equation}
\label{19}
Tr\left\{\bf{P_2P_2}^{\dag}\right\} \leq 1
\end{equation}

Lets do an eigen value decomposition such that:
\begin{equation}
\label{20}
\bf{H_1H_1^{\dag}}=\bf{U_{H_1H_1^{\dag}}}\bf{\Lambda_{H_1H_1^{\dag}}}\bf{U_{H_1H_1^{\dag}}^{\dag}},
\end{equation}
\begin{equation}
\label{21}
\bf{H_2H_2^{\dag}}=\bf{U_{H_2H_2^{\dag}}}\bf{\Lambda_{H_2H_2^{\dag}}}\bf{U_{H_2H_2^{\dag}}^{\dag}}
\end{equation}

Let:
\begin{equation}
\label{22}
\bf{\tilde{P_1}}= \bf{U_{H_1H_1^{\dag}}^{\dag}}\bf{P_1}
\end{equation}
\begin{equation}                                                                                
\label{23}
\bf{\tilde{P_2}}= \bf{U_{H_2H_2^{\dag}}^{\dag}}\bf{P_2}            
\end{equation}
Let $\bf{Z_1}~\succeq~0$ and $\bf{Z_2}~\succeq~0$ are positive semi-definite, such that:
\begin{equation}
\label{24}
\bf{Z_1}=\bf{\tilde{P_1}}\bf{\tilde{P_1}}^{\dag}
\end{equation}
\begin{equation}
\label{25}
\bf{Z_2}=\bf{\tilde{P_2}}\bf{\tilde{P_2}}^{\dag}
\end{equation}
Substitute \eqref{20} to \eqref{25} into \eqref{17},~\eqref{18}, and~\eqref{19}, the optimization problem can be re-written as:
\begin{equation}
\label{26}
max~~~Tr\left\{\bf{Z_1}\bf{\Lambda_{H_1H_1^{\dag}}}\right\}.snr + Tr\left\{{\bf{Z_2}}{\bf{\Lambda_{H_2H_2^{\dag}}}}\right\}.snr
\end{equation}
Subject to:
\begin{equation}
\label{27}
Tr\left\{\bf{Z_1}\right\} \leq 1
\end{equation}
\begin{equation}
\label{28}
Tr\left\{\bf{Z_2}\right\} \leq 1
\end{equation}
Solving the KKT conditions for the Lagrangian of the objective function \eqref{26} subject to \eqref{27} and \eqref{28} leads to:
\begin{equation}
\label{29}
{\bf{Z_1}}=\lambda^{-1}{\bf{\Lambda_{H_1H_1^{\dag}}}}.{snr},
\end{equation}
\begin{equation}
\label{30}
{\bf{Z_2}}=\lambda^{-1}{\bf{\Lambda_{H_2H_2^{\dag}}}}.{snr}
\end{equation}
Where $\lambda$ is the Lagrange multiplier of the optimization problem. The results in \eqref{29} and \eqref{30} prove that the optimal precoders in the low$-$snr perform mainly two operations: Firstly, it aligns the transmit directions with the eigenvectors of each sub-channel. Secondly, it performs power allocation over the sub-channels. 

\section{Optimal Power Allocation with Arbitrary Inputs}
\normalsize
To derive the optimal power allocation for MAC Gaussian channels driven by arbitrary inputs, we need to redefine the same objective function subject to the per user power constraints. The power allocation is a special case of the precoding which inherently includes a power allocation operation. However, the first reshapes the constellation by rotating or redefining the minimum distance between the points, while the second redistributes the power over the channels to guarantee a better utilization of the channel under the defined design criterion. Therefore, the power allocation problem
can be cast as a constrained non-linear optimization problem as follows:
\begin{equation}
\label{31}
max~~I(\bf{x_1, x_2; y})
\end{equation}
Subject to: 
\begin{equation}
\label{32}
{\bf{(P_1)}_{ij}}={\sqrt{p_{1j}}},~~ {\bf{(P_2)}_{ij}}={\sqrt{p_{2j}}}
\end{equation}
\begin{equation}
\label{33}
\sum\limits_{j=1}^{n_t} {{p}_{1j}} \leq 1,~ {{p}_{1j}}~\geq~0~~\forall j=1,...,n_t
\end{equation}
\begin{equation}
\label{34}
\sum\limits_{j=1}^{n_t} {{p}_{2j}} \leq 1,~ {{p}_{2j}}~\geq~0~~\forall j=1,...,n_t
\end{equation}
where ${\bf{p_{1}}}=diag({\sqrt{p_{1,1}}}$,$~...~$,${\sqrt{p_{1,n_t}}})$, and ${\bf{p_{2}}}=diag({\sqrt{p_{2,1}}}$,$~...~$,${\sqrt{p_{2,n_t}}})$. Capitalizing on the relation between the gradient of the mutual information with respect to the power allocation matrix and the MMSE, we provide the optimal power allocation in the following theorem.

\begin{theorem}
\label{theorem7.4}
The optimal power allocation that solves \eqref{31} subject to \eqref{32} to \eqref{34} satisfies:
\begin{equation}
\label{eq.35.4}
{{\bf{p}}_{1}}^{\star}=\gamma_1^{-1}\left({\bf{P_1}^{\star}}{\bf{H}}_{1}^{\dag}{\bf{H}_{1}}{\bf{P}_1}^{\star}{\bf{E}_{1}}-{\bf{P_1}^{\star}}{\bf{H}}_{1}^{\dag}{\bf{H}_{2}}{\bf{P}_2}^{\star}\bf{\mathbb{E}_y[\widehat{x}_2\widehat{x}_1^{\dag}]}\right)_{1j}
\end{equation}
\begin{equation}
\label{eq.36.4}
{{\bf{p}}_{2}}^{\star}=\gamma_2^{-1}\left({\bf{P_2}^{\star}}{\bf{H}}_{2}^{\dag}{\bf{H}_{2}}{\bf{P}_2}^{\star}{\bf{E}_{2}}-{\bf{P_2}^{\star}}{\bf{H}}_{2}^{\dag}{\bf{H}_{1}}{\bf{P}_1}^{\star}\bf{\mathbb{E}_y[\widehat{x}_1\widehat{x}_2^{\dag}]}\right)_{2j}
\end{equation}
\end{theorem}
Where $\gamma_1$ and $\gamma_2$ are the Lagrange multipliers normalized by the received $snr$, of the optimization problem in \eqref{31}. 
\begin{proof}
The proof of Theorem~\ref{theorem7.4} relies on the KKT conditions \cite{52} and the relation between the gradient of the mutual information with respect to the precoding matrices and the MMSE, see Appendix H.
\end{proof}

For the special case where we have a single channel per user in the two-user MAC, \eqref{eq.35.4} and \eqref{eq.36.4} can be simplified to the following:
\begin{equation}
\label{37}
{\bf{{p}}_{1}}^{\star}=\gamma_1^{-1}{{p_{1j}}^{\star}}|{{h_1}_{j}}|^{2}{\mathbb{E}_y[({\bf{x_1-\widehat{x}_1}})({\bf{x_1-\widehat{x}_1}})^{\dag}]}
-\gamma_1^{-1}{{p_{2j}}^{\star}}|{{h_1}_{j}}^{*}{{h_2}_{j}}|\bf{\mathbb{E}_y[({\bf{x_1-\widehat{x}_1}})({\bf{x_2-\widehat{x}_2}})^{\dag}]}
\end{equation}
\begin{equation}
\label{38}
{\bf{{p}}_{2}}^{\star}=\gamma_2^{-1}{{p_{2j}}^{\star}}|{{h_2}_{j}}|^{2}{\mathbb{E}_y[({\bf{x_2-\widehat{x}_2}})({\bf{x_2-\widehat{x}_2}})^{\dag}]}
-\gamma_2^{-1}{{p_{1j}}}^{\star}|{{h_2}_{j}}^{*}{{h_1}_{j}}|{\mathbb{E}_y[({\bf{x_2-\widehat{x}_2}})({\bf{x_1-\widehat{x}_1}})^{\dag}]}
\end{equation}

\normalsize

In \eqref{37} and \eqref{38}, due to the fact that the MSE is less than or equal to 1, and consider the special case where the second covariance term is zero, i.e., no user is interfering with the other, we can express the optimal power allocation with the single user mercury/waterfilling form \cite{39} as follows:
\begin{equation}
\left\{\begin{array}{l l}
 {\bf{p_1^{\star}}}={\frac{1}{snr|{h_1}_{j}|^{2}}mmse_1^{-1}}\left(\frac{\gamma_1}{|{h_1}_{j}|^{2}}\right) &\mbox{${\gamma_1 < |{h_1}_{j}|^{2}}$} \\
 {\bf{p_1^{\star}}}=0,    &\mbox{${\gamma_1 \geq |{h_1}_{j}|^{2}}$} \\
\end{array} \right.
\label{eq:39}
\end{equation}

\begin{equation}
\left\{\begin{array}{l l}
{\bf{p_2^{\star}}}=\frac{1}{snr|{h_2}_{j}|^{2}}{mmse_2^{-1}}\left(\frac{\gamma_2}{|{h_2}_{j}|^{2}}\right) &\mbox{${\gamma_2 < |{h_2}_{j}|^{2}}$}\\
{\bf{p_2^{\star}}}=0,    &\mbox{${\gamma_2 \geq |{h_2}_{j}|^{2}}$}\\
\end{array} \right.
\label{eq:40}
\end{equation}

There is a unique set $\bf{p_1^{\star}}$, $\bf{p_2^{\star}}$ that satisfy the KKT conditions when the problem is strictly concave, corresponding to the global maximum.

\section{Numerical Results}
In this section we will introduce a set of illustrative results of the two-user MAC with arbitrary inputs. In particular, we will present the simplified case where the two user channels are transmitting over scalar channels or single channel per user. We used Monte-Carlo method to generate the results.

Figure~\ref{fig:Figure1} illustrates the decay in the mutual information at the $45^{\small o}$ when the two BPSK inputs cancel each other as they lie in the null space of the channel; called the Voronoi region, with the channels are such that $h_1=1$, and $h_2=1$. 

If the channel gains or the total powers are different, the decay will be shifted into another line in the coordinate axis. However, we can easily check that if we induce orthogonality, for example if $h_1=1$ and $h_2=1j$, the mutual information will be the maximum achievable one without this decay region. 

\begin{figure}[ht!]
    \begin{center}
        \mbox{\includegraphics[width=6.4in,height=4in]{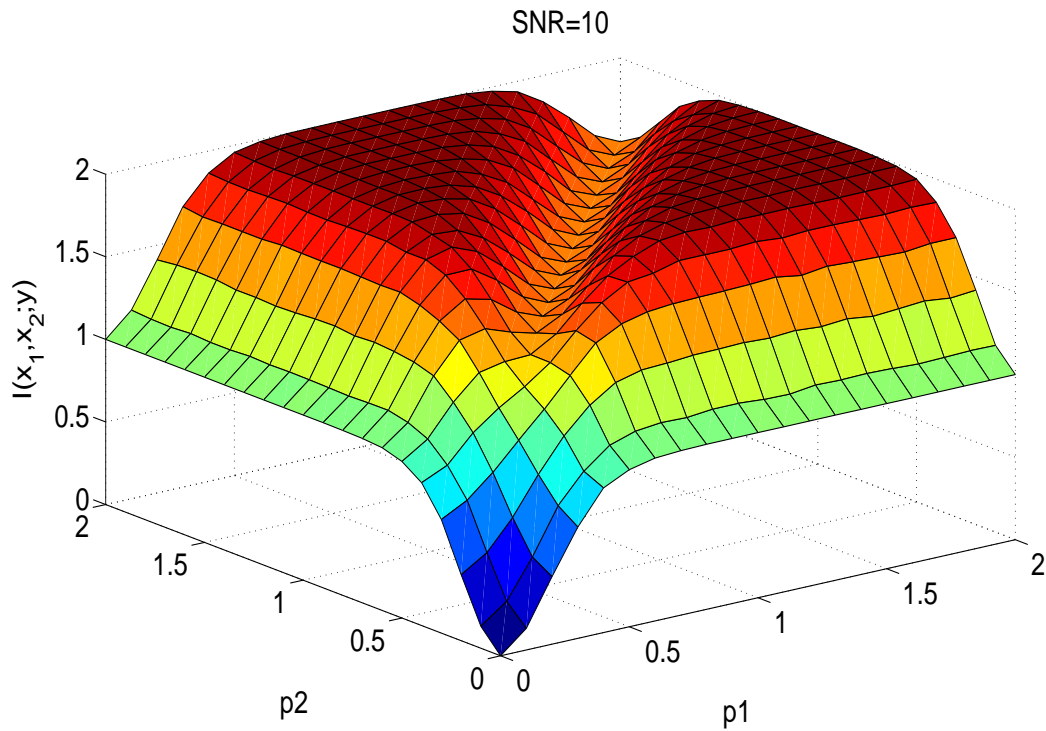}}
    \caption{The two-user MAC mutual information with BPSK inputs.}
    \label{fig:Figure1}
    \end{center}
    \label{Figure1}
    \end{figure}
\begin{figure}[ht!]
    \begin{center}
        \mbox{\includegraphics[width=6.4in,height=4.1in]{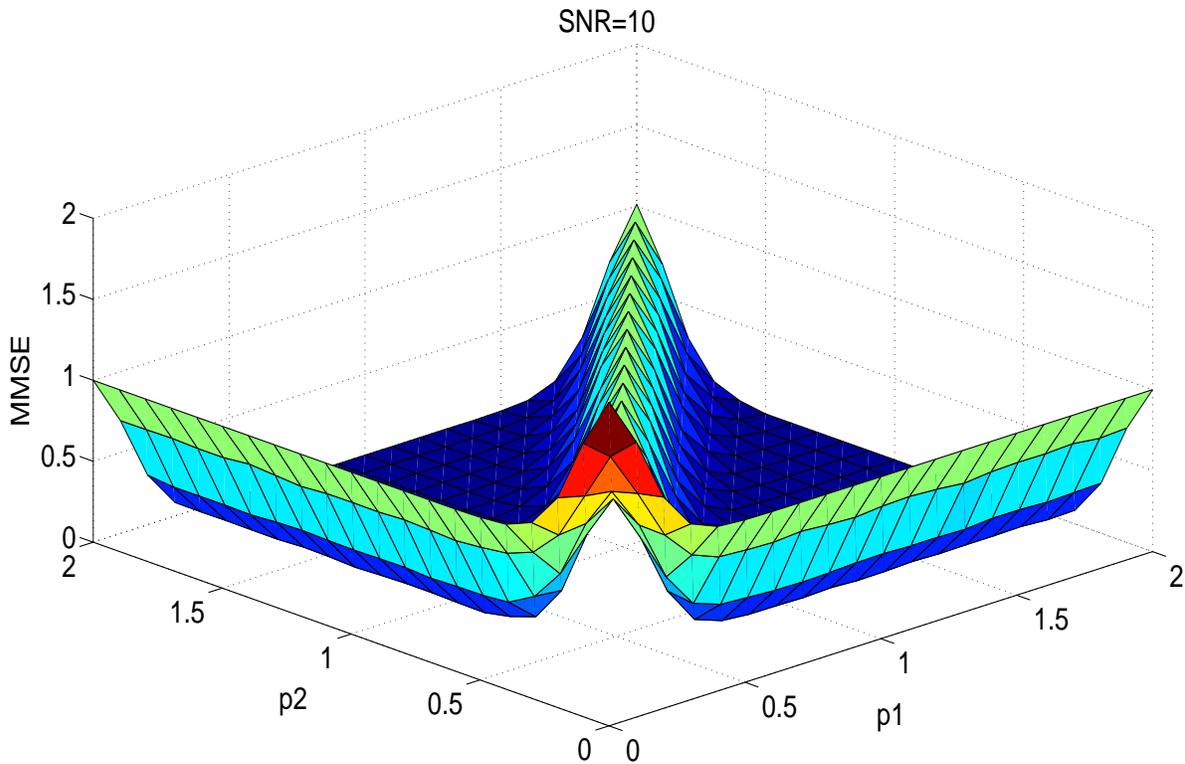}}
    \caption{The two-user MAC MMSE with BPSK inputs.}
    \label{fig:Figure2}
    \end{center}
    \label{Figure2}
\end{figure}

In Figure~\ref{fig:Figure2} we show the total MMSE for the two-user MAC. It is interesting to notice that the MMSE behavior corresponds to swapping the mutual information curve in Figure~\ref{fig:Figure1}. This can let us visualize the relation between the gradient of the mutual information and the MMSE \cite{35}, \cite{34}. 

Of particular relevance to observe that the decay in the mutual information saturates at 1.5 bits while the decay of the MMSE saturates at 1 for equal power of both user's, e.g., $(p_1,p_2)=(2,2)$ in Figure~\ref{fig:Figure1}, and Figure~\ref{fig:Figure2}, respectively.

This is explained in the new fundamental relation which quantifies a loss of 0.5 bits in the data rates. For the case of Gaussian inputs the relation still holds, but the decay encountered in the arbitrary input setup will not exist. This can be easily verified theoretically and by simulation. For further analysis on Gaussian and mixed inputs, we refer the reader to~\cite{SamahLaura}.

Notice that the rate loss in the first user mutual information is due to the interference incurred by the second user. This can be understood via the negative second term in the gradient of the mutual information, similarly, for the other user. This term is a function of the main channel, the interferer channel, the main and interferer powers, and the input estimates. 

Therefore, the interference from user 1 to user 2 can be interpreted as:
\begin{equation}
\label{41}
\bf{H_{1}}^{\dag}{\bf{H}_{2}}{\bf{P}_2}^{\star}\bf{\mathbb{E}[\widehat{x}_2\widehat{x}_1^{\dag}]},
\end{equation}
where we can re-write equation \eqref{41} for the single channel per user case in terms of the covariance of the interferer as follows:
\begin{equation}
\label{42}
\frac{1}{snr~{h_{1}^{\dag}}{{h}_{2}}}cov(snr~h_{1}^{\dag}{{h}_{2}}p_2^{\star}),
\end{equation}

Similarly, the interference from user 2 to user 1 can be interpreted as:
\begin{equation}
\label{43}
\bf{H_{2}^{\dag}}{\bf{H}_{1}}{\bf{P}_1}^{\star}\bf{\mathbb{E}[\widehat{x}_1\widehat{x}_2^{\dag}]},
\end{equation}
Where we can re-write equation \eqref{43} for the single channel per user case in terms of the covariance of the interferer as follows:
\begin{equation}
\label{44}
\frac{1}{snr~{h_{2}^{\dag}}{{h}_{1}}}cov(snr~{h_{2}^{\dag}}{{h}_{1}}p_1^{\star}).
\end{equation}

Figure~\ref{fig:Figure3} and Figure~\ref{fig:Figure4} illustrate the MMSE per user in the two-user MAC channel. MMSE1 and MMSE2 correspond to the mean-squared error of the first and second user respectively. It is worth to note that the sum of the two MMSEs leads to the total MMSE in Figure~\ref{fig:Figure2}. 

Figure~\ref{fig:Figure5} illustrates the covariance of one interferer into the other user. We can see the negative term that leads to the loss incurred in the achieved mutual information.

Finally, Figure~\ref{fig:Figure6} presents the optimal power allocation for the two-user MAC. When the total power for the first user is larger than that for the second user, the optimal power allocation for the first will be to allocate its total power, as far as the second user doesn't extremely interfere. 

However, for the second user, the optimal power allocation will start with an allocation of the total power, then iteratively optimize it by decreasing it. Therefore, precoding is of great importance to such cases where we can align the transmit directions - or caused interference - and maximize the information rates.
\begin{figure}[ht!]
    \begin{center}
        \mbox{\includegraphics[width=6.6in,height=3.8in]{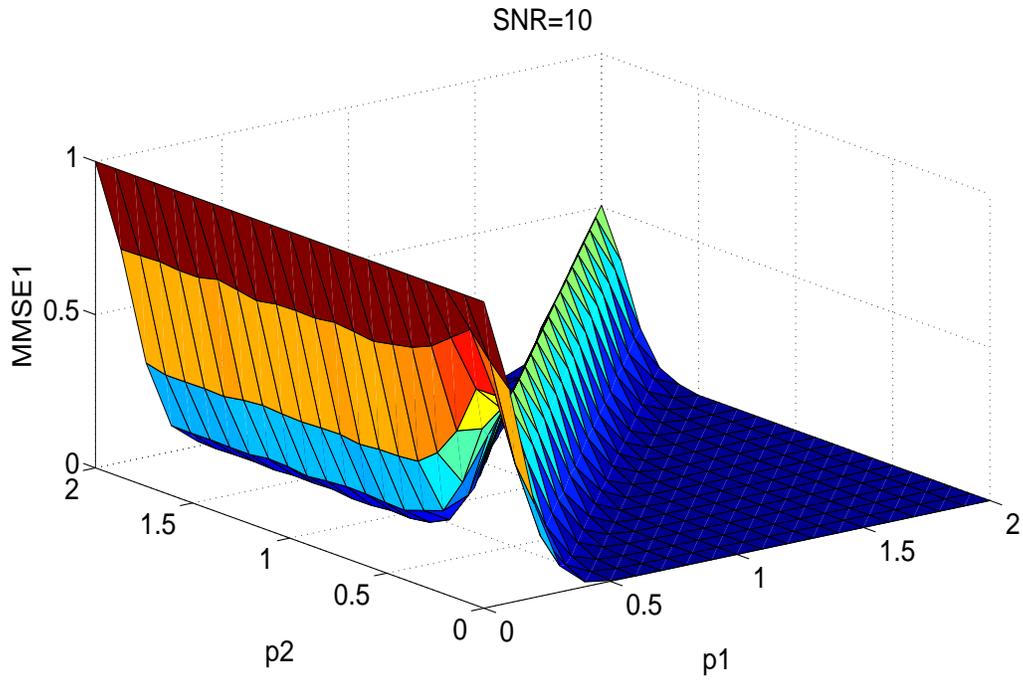}}
    \caption{MMSE1 for the first user of the two-user MAC with BPSK input.}
    \label{fig:Figure3}
    \end{center}
    \label{Figure3}
   \end{figure}
\begin{figure}[ht!]
    \begin{center}
        \mbox{\includegraphics[width=6.6in,height=3.8in]{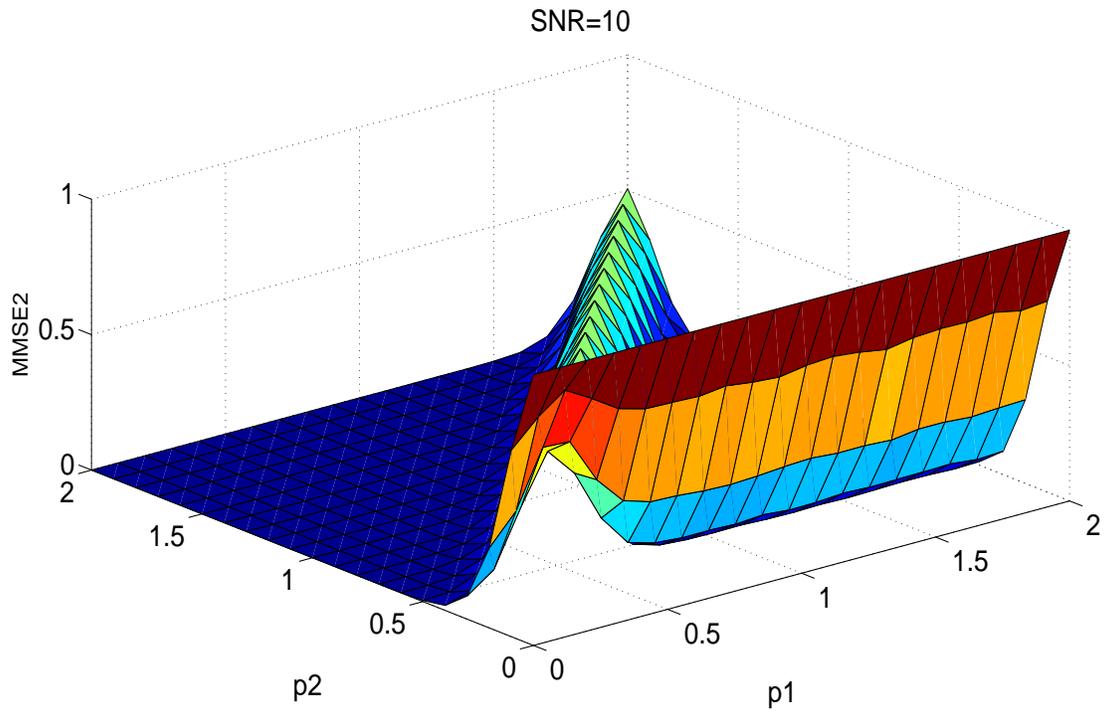}}
    \caption{MMSE2 for the second user of the two-user MAC with BPSK input.}
    \label{fig:Figure4}
    \end{center}
    \label{Figure4}
\end{figure}

\begin{figure}[ht!]
    \begin{center}
        \mbox{\includegraphics[width=6.6in,height=3.9in]{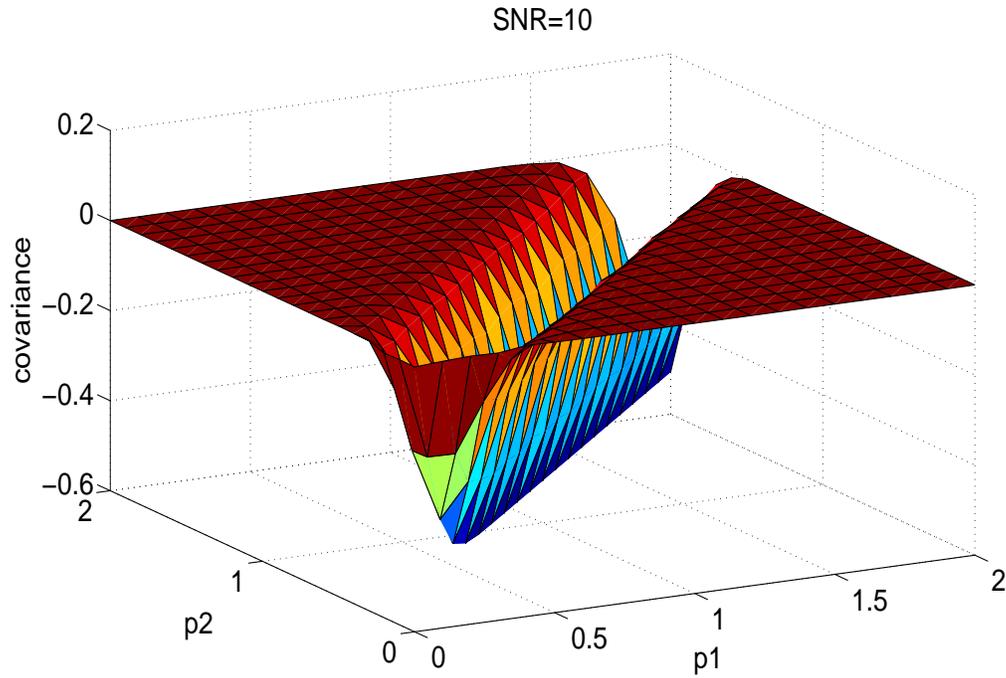}}
    \caption{Covariance for the first user of the two-user MAC with BPSK input.}
    \label{fig:Figure5}
    \end{center}
    \label{Figure5}
\end{figure}
\begin{figure}[ht!]
    \begin{center}
        \mbox{\includegraphics[width=6.6in,height=3.9in]{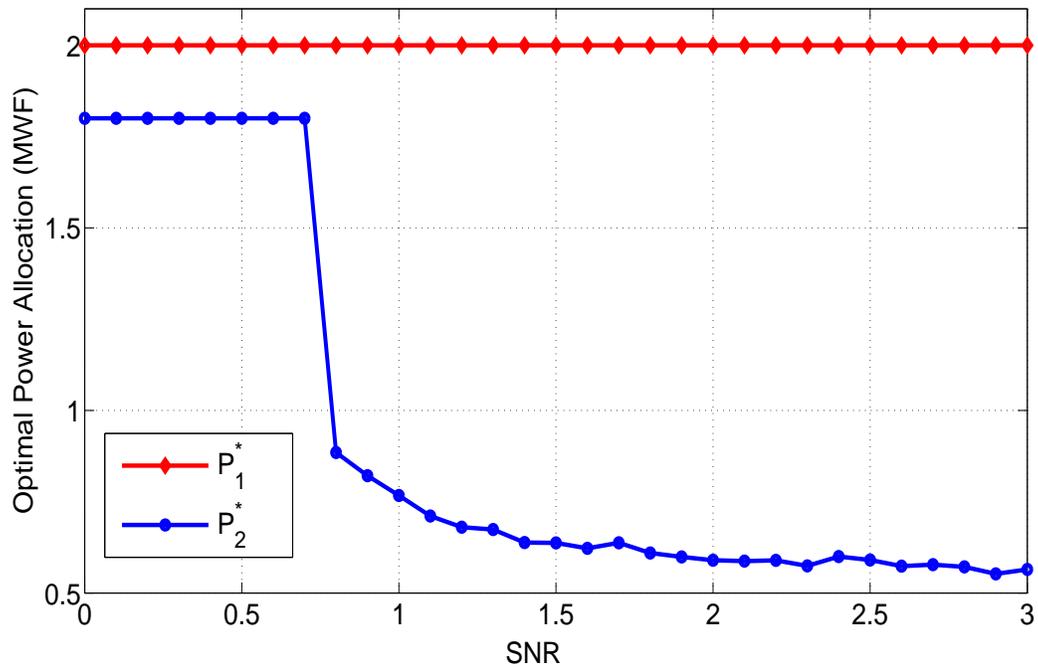}}
    \caption{Optimal power allocation of the two user MAC with BPSK inputs.}
    \label{fig:Figure6}
    \end{center}
    \label{Figure6}
\end{figure}

\section{Conclusions}
In this paper, we studied the two-user MAC Gaussian channels with arbitrary inputs. We provide a generlization of the fundamental relation between the mutual information and the MMSE to a new fundamental relation which applies to multiuser channel setups. Further, we proved our generlized framework deriving the relation between the gradient of the mutual information and the MMSE for any input setup. We build upon derivations for the optimal precoding and optimal power allocation. We specialize our results to different cases of the input users, where for the simple case when each user transmits over a SISO channel, the relations can be mapped to the ones by Guo et al. for point to point channels with general distributions \cite{34}, and by Palomar et al. for the linear vector Gaussian channels with mutually interfering inputs \cite{35}. We have shown that the optimal precoder can follow a fixed point equation similar to \cite{3} under certain conditions, where we identified the optimal precoder structure with respect to the power allocation, and the right and left signular vectors. Much simplification to the setup where both inputs don't interfere with each other leads to the mercury/waterfilling interpretation for the single user case with binary inputs \cite{39}. This casts further insights to have an interpretation of the interference as discussed. Additionally, it casts further research questions in quantifying the losses incurred in the mutual information due to interference with and without precoding, and in conjunction with error correcting codes. Furthermore, we have analyzed the system at the low-SNR regime, we have derived novel closed forms of the MMSE and the mutual information at the low SNR based on our new unveiled fundamental relation. The interference at the low-SNR adds non-linear terms at the low SNR which decides on how fast the capacity can be acheived with respect to the SNR.

Of particular relevance are the implications of our derived new fundamental relation in the study of the fundamental limits of cooperation, if a multi-cell processing framework is to be implemented between the BSs in a network. In particular, the interpretation of the interference in the studied two-user MAC, which is basically extracted from the gradient formula, can be extended to the $k$-user MAC channel where the impact of interference incurs higher losses into the information rates. Therefore, this have direct impact on the optimal designs of power allocation and precoding. In fact, a global CSI and/or data sharing will not be possible, and this enforces a framework where muli-cell processing (MCP) should be implemented in clusters of known size. Furthermore, it is straightforward to think about setups where the cooperation and a clustered knowledge, i.e., a MCP with different levels of cooperation, can provide robust techniques to cancel or decode the interference at the transmitter/receiver sides and so improving the network spectral efficiency. In addition, its of particular relevance to understand the significance of our derivations on setups that include higher number of users accessing the network. For instance, the application of such formulas into scenarios where main users and cognitive users co-exist. The main users can design there optimal power allocation and optimal precoding subject to, the known interferer(s), i.e., the cognitive user interference, whether, they will be transmitting over SISO channels, or using OFDM modulation. Further, we have also she light on a new class of problems where a multi-tier precoding could exist, one example is massive MIMOs. In particular, a new formula can be utilized to express the signal to interference and noise ratio (SINR) based on the interpretation of the interference.

\section{Appendix A: Proof of Theorem~\ref{theorem1.0}}
The conditional probability density for the two-user MAC can be written as follows:
\begin{equation}
p_{y|x_1,x_2}({\bf{y|x_1, x_2}})=\frac{1}{\pi^{n_r}}e^{-\left\|y-\sqrt{snr}{\bf{{H_1P_1x_1}}}-\sqrt{snr}{\bf{H_2P_2x_2}}\right\|^{2}}
\end{equation}
Thus, the corresponding mutual information is:
\begin{equation}
I({\bf{x_1, x_2;y}})=\mathbb{E}\left[log \left(\frac{p_{y|x_1,x_2}({\bf{y|x_1, x_2}})}{p_{y}({\bf{y}})}\right)\right]
\end{equation}
\begin{equation}
I({\bf{x_1, x_2;y}})=-n_r{{log(\pi e)}} -\mathbb{E}\left[log\left(p_{y}({\bf{y}})\right)\right]
\end{equation}
\begin{equation}
I({\bf{x_1, x_2;y}})=-n_r{{log(\pi e)}} -  \int p_{y}({\bf{y}})log\left(p_{y}({\bf{y}})\right) d\bf{y}
\end{equation}
Then, the derivative of the mutual information with respect to the SNR is as follows:
\begin{equation}
\frac{d I({\bf{x_1, x_2;y}})}{d snr}=-\frac{\partial}{\partial snr} \int {p_{y}}({\bf{y}})log \left(p_{y}({\bf{y}})\right) d\bf{y}
\end{equation}
\begin{equation}
~~=-\int \left({p_{y}}({\bf{y}})\frac{1}{{p_{y}}({\bf{y}})}+log \left(p_{y}({\bf{y}})\right)\right)\frac{\partial {p_{y}}({\bf{y}})}{\partial snr} d\bf{y}
\end{equation}
\begin{equation}
\label{45.0}
~~=-\int \left(1+\bf{log\left(p_{y}({\bf{y}})\right)}\right)\frac{\partial {p_{y}}({\bf{y}})}{\partial snr} d\bf{y}
\end{equation}

Where the probability density function of the received vector $\bf{y}$ is given by:
\begin{equation}
p_{y}({\bf{y}})=\sum_{{\bf{x_{1},x_{2}}}} p_{y|x_{1},x_{2}}({\bf{y|x_{1},x_{2}}})p_{x_1,x_2}({\bf{x_{1}}},{\bf{x_{2}}})
\end{equation}
\begin{equation}
~~=\mathbb{E}_{x_1,x_2}\left[{p_{y|x_{1},x_{2}}}({\bf{y|x_{1},x_{2}}})\right]
\end{equation}

The derivative of the conditional output with respect to the SNR can be written as:
\begin{multline}
\frac{\partial {{p_{y|x_{1},x_{2}}}}({\bf{y|x_{1},x_{2}}})}{\partial snr}=\\
-{{p_{y|x_{1},x_{2}}}}({\bf{y|x_{1},x_{2}}})\frac{\partial }{\partial snr}\left({\bf{y}}-\sqrt{snr}{\bf{H_1P_1x_1}}-\sqrt{snr}{\bf{H_2P_2x_2}}\right)^{\dag} \times \\
\left({\bf{y}}-\sqrt{snr}{\bf{H_1P_1x_1}}-\sqrt{snr}{\bf{H_2P_2x_2}}\right)
\end{multline}
\begin{multline}
~~=-\frac{1}{\sqrt{snr}}\left(({\bf{H_1P_1x_1}})^\dag - ({\bf{H_2P_2x_2}})^\dag \right) \left({\bf{y}}-\sqrt{snr}{\bf{H_1P_1x_1}}-\sqrt{snr}{\bf{H_2P_2x_2}}\right) \times \\
{p_{y|x_{1},x_{2}}}({\bf{y|x_{1},x_{2}}})
\end{multline}
\begin{equation}
~~=-\frac{1}{\sqrt{snr}}\left(({\bf{H_1P_1x_1}})^\dag-({\bf{H_2P_2x_2}})^\dag \right){\bf{\nabla_{y}}}{p_{y|x_{1},x_{2}}}({\bf{y|x_{1},x_{2}}})
\end{equation}

Therefore, we have:
\begin{multline}
\label{46.0}
\mathbb{E}_{x_1,x_2}\left[{\nabla_{snr}}{p_{y|x_{1},x_{2}}}({\bf{y|x_{1},x_{2}}})\right]=\\
\mathbb{E}_{x_1,x_2}\left[-\frac{1}{\sqrt{snr}}\left(({\bf{H_1P_1x_1}})^\dag-({\bf{H_2P_2x_2}})^\dag \right)\nabla_{\bf{y}}p_{y|x_{1},x_{2}}({\bf{y|x_{1},x_{2}}})\right]
\end{multline}
Substitute \eqref{46.0} into \eqref{45.0}, we get:
\begin{multline}
\frac{d I({\bf{x_1, x_2;y}})}{d snr}=\frac{1}{\sqrt{snr}} \int \left(1+{log}\left({p_{y}}({\bf{y}})\right)\right)\mathbb{E}_{x_1,x_2}[\left(({\bf{H_1P_1x_1}})^\dag-({\bf{H_2P_2x_2}})^\dag \right) \times \\
\bf{\nabla_{y}}\bf{p_{y|x_{1},x_{2}}}(\bf{y|x_{1},x_{2}})] d\bf{y}
\end{multline}
\begin{equation}
~~=\frac{1}{\sqrt{snr}}\mathbb{E}_{x_1,x_2}\left[\left(\int \left(1+{log}\left({p_{y}}({\bf{y}})\right)\right)\left(({\bf{H_1P_1x_1}})^\dag-({\bf{H_2P_2x_2}})^\dag \right)\nabla_{\bf{y}}{p_{y|x_{1},x_{2}}}({\bf{y|x_{1},x_{2}}})d\bf{y}\right) \right]
\end{equation}

Using integration by parts applied to the real and imaginary parts of $\bf{y}$ we have:
\begin{multline}
\label{47.0}
\int \left(1+{log}\left({p_{y}}({\bf{y}})\right)\right) \frac{\partial p_{y|x_{1},x_{2}}({\bf{y|x_{1},x_{2}}})}{\partial t} d{\bf{t}}= \\
\int \left(1+{log}\left({p_{y}}({\bf{y}})\right)\right) p_{y|x_{1},x_{2}}({\bf{y|x_{1},x_{2}}}){|}_{-\infty}^{\infty}-\int_{-\infty}^{\infty} \frac{1}{p_{y}({\bf{y}})} \frac{\partial {p_{y}}({\bf{y}})}{\partial t}{p_{y|x_{1},x_{2}}}({\bf{y|x_{1},x_{2}}})d\bf{t}
\end{multline}

The first term in \eqref{47.0} goes to zero as $\left\|\bf{y}\right\| \rightarrow \infty$. Therefore,
\begin{equation}
\frac{d I({\bf{x_1, x_2;y}})}{d snr}=\frac{1}{\sqrt{snr}}\mathbb{E}_{x_1,x_2}\left[-\int \left( \left(({\bf{H_1P_1x_1}})^\dag-({\bf{H_2P_2x_2}})^\dag \right)\frac{p_{y|x_{1},x_{2}}({\bf{y|x_{1},x_{2}}})}{p_{y}({\bf{y}})}\nabla_{\bf{y}}p_{y}({\bf{y}})d\bf{y} \right) \right]
\end{equation}
\begin{equation}
\frac{d I({\bf{x_1, x_2;y}})}{d snr}=-\frac{1}{\sqrt{snr}}\int {\nabla_{\bf{y}}}{p_{y}({\bf{y}})}\mathbb{E}_{x_1,x_2}\left[ \left(({\bf{H_1P_1x_1}})^\dag-({\bf{H_2P_2x_2}})^\dag \right) \frac{{p_{y|x_{1},x_{2}}}({\bf{y|x_{1},x_{2}}})}{{p_{y}({\bf{y}})}}\right]d\bf{y}
\end{equation}
\begin{equation}
\label{48.0}
\frac{d I({\bf{x_1, x_2;y}})}{d snr}=-\frac{1}{\sqrt{snr}}\int \nabla_{\bf{y}}{{p_{y}({\bf{y}})}}\mathbb{E}_{x_1,x_2}\left(({\bf{H_1P_1}})^\dag \mathbb{E}_{x_1|y}\left[{\bf{x_1|y}}\right]^{\dag}-({\bf{H_2P_2}})^\dag \mathbb{E}_{x_2|y}\left[{\bf{x_2|y}}\right]^{\dag} \right) d\bf{y}
\end{equation}

However,
\begin{multline}
\label{49.0}
\nabla_{\bf{y}}{p_{y}(\bf{y})}=\nabla_{\bf{y}}\mathbb{E}_{x_1,x_2}\left[{p_{y|x_{1},x_{2}}}({\bf{y|x_{1},x_{2}}})\right] \\
=\mathbb{E}_{x_1,x_2}\left[\nabla_{\bf{y}}{p_{y|x_{1},x_{2}}}({\bf{y|x_{1},x_{2}}})\right] \\
=-\mathbb{E}_{x_1,x_2}\left[{p_{y|x_{1},x_{2}}}({\bf{y|x_{1},x_{2}}})\left({\bf{y}}-\sqrt{snr}{\bf{H_1P_1x_1}}-\sqrt{snr}{\bf{H_2P_2x_2}}\right)\right]\\
=-\mathbb{E}_{x_1,x_2}\left[{p_{y}}({\bf{y}})\left({\bf{y}}-\sqrt{snr}{\bf{H_1P_1x_1}}-\sqrt{snr}{\bf{H_2P_2x_2}}\right)|{\bf{y}}\right] \\
=-{p_{y}}({\bf{y}})\left({\bf{y}}-\sqrt{snr}{\bf{H_1P_1}}\mathbb{E}_{x_1|y}[{\bf{x_1|y}}]-\sqrt{snr}{\bf{H_2P_2}}\mathbb{E}_{x_2|y}[{\bf{x_2|y}}]\right)\\
\end{multline}

Substitute \eqref{49.0} into \eqref{48.0} we get:
\begin{align}
\frac{d I({\bf{x_1, x_2;y}})}{d snr}=\frac{1}{\sqrt{snr}}\int p_{y}({\bf{y}})\left({\bf{y}}-\sqrt{snr}{\bf{H_1P_1}}\mathbb{E}_{x_1|y}[{\bf{x_1|y}}]-\sqrt{snr}{\bf{H_2P_2}} \mathbb{E}_{x_2|y}[{\bf{x_2|y}}]\right) \times \nonumber \\
\mathbb{E}_{x_1,x_2}\left(({\bf{H_1P_1}})^\dag \mathbb{E}_{x_1|y}\left[{\bf{x_1|y}}\right]^{\dag}-({\bf{H_2P_2}})^\dag \mathbb{E}_{x_2|y}\left[{\bf{x_2|y}}\right]^{\dag} \right) d\bf{y}
\end{align}
\begin{align}
\frac{d \bf{I(x_1, x_2;y)}}{d snr}=\frac{1}{\bf{\sqrt{snr}}}\mathbb{E}_y[\bf{yx_1^{\dag}}]\bf{(H_1P_1)^{\dag}}-\frac{1}{\sqrt{snr}}\mathbb{E}_y[\bf{yx_2^{\dag}}]\bf{(H_2P_2)^{\dag}} \nonumber \\
-\mathbb{E}_y[{\bf{H_1P_1}}\mathbb{E}_{x_1|y}[{\bf{x_1|y}}]\mathbb{E}_{x_1|y}[{\bf{x_1|y}}]^{\dag}{\bf{(H_1P_1)^{\dag}}}+\mathbb{E}_y[{\bf{H_1P_1}}\mathbb{E}_{x_1|y}[{\bf{x_1|y}}]\mathbb{E}_{x_2|y}[{\bf{x_2|y}}]^{\dag}]{\bf{(H_2P_2)^{\dag}}} \nonumber \\
+\mathbb{E}_y[{\bf{H_2P_2}}\mathbb{E}_{x_2|y}[{\bf{x_2|y}}]\mathbb{E}_{x_2|y}[{\bf{x_2|y}}]^{\dag}{\bf{(H_2P_2)^{\dag}}}-\mathbb{E}_y[{\bf{H_2P_2}}\mathbb{E}_{x_2|y}[{\bf{x_2|y}}]\mathbb{E}_{x_1|y}[{\bf{x_1|y}}]^{\dag}]{\bf{(H_1P_1)^{\dag}}}
\end{align}

Therefore,
\begin{multline}
\frac{d I({\bf{x_1, x_2;y}})}{d snr}={\bf{H_1P_1}}\mathbb{E}_{x_1}[{\bf{x_1x_1^{\dag}}}]{\bf{(H_1P_1)^{\dag}}} \nonumber \\
-{\bf{H_1P_1}}\mathbb{E}_y[\mathbb{E}_{x_1|y}[{\bf{x_1|y}}]\mathbb{E}_{x_1|y}[{\bf{x_1|y}}]^{\dag}{\bf{(H_1P_1)^{\dag}}} \nonumber \\
+{\bf{H_1P_1}}\mathbb{E}_y[\mathbb{E}_{x_1|y}[\bf{x_1|y}]\mathbb{E}_{x_2|y}[{\bf{x_2|y}}]^{\dag}]{\bf{(H_2P_2)^{\dag}}}\nonumber \\
-{\bf{H_2P_2}}\mathbb{E}_{x_2}[{\bf{x_2x_2^{\dag}}}]{\bf{(H_2P_2)^{\dag}}} \nonumber \\
+{\bf{H_2P_2}}\mathbb{E}_y[\mathbb{E}_{x_2|y}[{\bf{x_2|y}}]\mathbb{E}_{x_2|y}[{\bf{x_2|y}}]^{\dag}{\bf{(H_2P_2)^{\dag}}}  \nonumber \\
-{\bf{H_2P_2}}\mathbb{E}_y[\mathbb{E}_{x_1|y}[{\bf{x_1|y}}]\mathbb{E}_{x_2|y}[{\bf{x_2|y}}]^{\dag}]{\bf{(H_1P_1)^{\dag}}}
\end{multline}
\begin{multline}
\frac{d I({\bf{x_1, x_2;y}})}{d snr}={\bf{H_1P_1}}{\bf{E_1}}{\bf{(H_1P_1)^{\dag}}}-{\bf{H_2P_2}}{\bf{E_2}}{\bf{(H_2P_2)^{\dag}}} \nonumber \\
+{\bf{H_1P_1}}\mathbb{E}_y[\mathbb{E}_{x_1|y}[{\bf{x_1|y}}]\mathbb{E}_{x_2|y}[{\bf{x_2|y}}]^{\dag}]{\bf{(H_2P_2)^{\dag}}} \nonumber \\
-{\bf{H_2P_2}}\mathbb{E}_y[\mathbb{E}_{x_2|y}[{\bf{x_2|y}}]\mathbb{E}_{x_1|y}[{\bf{x_1|y}}]^{\dag}]{\bf{(H_1P_1)^{\dag}}}
\end{multline}
Therefore, due to the fact that the expectation remains the same if ${\bf{y}} \sim \mathcal{CN}(\sqrt{snr}{\bf{H_1P_1x_1}}+\sqrt{snr}{\bf{H_2P_2x_2,I}})$ or ${\bf{y}} \sim \mathcal{CN}(-\sqrt{snr}{\bf{H_1P_1x_1}}-\sqrt{snr}{\bf{H_2P_2x_2,I}})$, therefore, $\mathbb{E}[{\bf{x_1x_1^\dag}}]=-\mathbb{E}[{\bf{x_1x_1^\dag}}]$ and $\mathbb{E}[\mathbb{E}[{\bf{x_1|y}}]\mathbb{E}[{\bf{x_1|y}}]^{\dag}]=-\mathbb{E}[\mathbb{E}[{\bf{x_1|y}}]\mathbb{E}[{\bf{x_1|y}}]^{\dag}]$, and similarly for the other input, the derivative of the mutual information with respect to the SNR and the per users mmse and input estimates (or covariances) is as follows:
\begin{multline}
\frac{d I({\bf{x_1, x_2;y}})}{d snr}=mmse_1(snr)+mmse_2(snr)+Tr\left\{{\bf{H_1P_1}}\mathbb{E}_{y}[{\bf{\widehat{x}_1{\widehat{x}_2}^{\dag}}}]{\bf{(H_2P_2)^{\dag}}}\right\}  \nonumber \\
-Tr\left\{{\bf{H_2P_2}}\mathbb{E}_y[{\bf{\widehat{x}_2{\widehat{x}_1}^{\dag}}}]{\bf{(H_1P_1)^{\dag}}}\right\}
\end{multline}
Therefore, we can write the derivative of the derivative of the mutual information with respect to the snr as follows:
\begin{equation}
\frac{d I(snr)}{d snr}=mmse(snr)+\psi(snr)
\end{equation}

Therefore, Theorem~\ref{theorem1.0} has been proved as a generalization of the one by Guo, Shamai, Verdu in \cite{34} to the multiuser case.
\section{Appendix B: Proof of Theorem~\ref{theorem1.4}}
\normalsize
The conditional probability density for the two-user MAC can be written as follows:
\begin{equation}
p_{y|x_1,x_2}({\bf{y|x_1, x_2}})=\frac{1}{\pi^{n_r}}e^{-\left\|y-\sqrt{snr}{\bf{H_1P_1x_1}}-\sqrt{snr}{\bf{H_2P_2x_2}}\right\|^{2}}
\end{equation}

Thus, the corresponding mutual information is:
\begin{equation}
I({\bf{x_1, x_2;y}})=\mathbb{E}\left[log\left(\frac{p_{y|x_1,x_2}({\bf{y|x_1, x_2}})}{p_{y}({\bf{y}})}\right)\right]
\end{equation}
\vspace{-0.6cm}
\begin{equation}
I({\bf{x_1, x_2;y}})=-n_r{{log(\pi e)}} -{\mathbb{E}\left[log\left(p_{y}({\bf{y}})\right)\right]}
\end{equation}

Then, the gradient of the mutual information with respect to the channel of user 1 on the two-user MAC is as follows:
\begin{equation}
\frac{\partial I({\bf{x_1, x_2;y}})}{\partial {\bf{H_1^{\dag}}}}=-\frac{\partial}{\partial {\bf{H_1^{\dag}}}} \int {p_{y}}({\bf{y}})log\left(p_{y}({\bf{y}})\right) d\bf{y}
\end{equation}

\begin{equation}
~~=-\int \left({p_{y}}({\bf{y}})\frac{1}{{p_{y}}({\bf{y}})}+log\left(p_{y}({\bf{y}})\right)\right)\frac{\partial {p_{y}}({\bf{y}})}{\partial {\bf{H_1^{\dag}}}} d\bf{y}
\end{equation}
\begin{equation}
\label{45}
~~=-\int \left(1+log\left(p_{y}({\bf{y}})\right)\right)\frac{\partial {p_{y}}({\bf{y}})}{\partial {\bf{H_1^{\dag}}}} d\bf{y}
\end{equation}

Where the probability density function of the received vector $\bf{y}$ is given by:
\begin{equation}
{p_{y}}({\bf{y}})=\sum_{{\bf{x_{1},x_{2}}}} p_{y|x_{1},x_{2}}({\bf{y|x_{1},x_{2}}})p_{x_1,x_2}({\bf{x_{1}}},{\bf{x_{2}}})
\end{equation}
\begin{equation}
~~=\mathbb{E}_{x_1,x_2}\left[{p_{y|x_{1},x_{2}}}({\bf{y|x_{1},x_{2}}})\right]
\end{equation}

The derivative of the conditional output can be written as:
\begin{multline}
\frac{\partial {p_{y|x_{1},x_{2}}}({\bf{y|x_{1},x_{2}}})}{\partial {\bf{H_1^{\dag}}}}=\\
-{p_{y|x_{1},x_{2}}}({\bf{y|x_{1},x_{2}}})\frac{\partial }{\partial {\bf{H_1^{\dag}}}}\left({\bf{y}}-\sqrt{snr}{\bf{H_1P_1x_1}}-\sqrt{snr}{\bf{H_2P_2x_2}}\right)^{\dag} \times \\
\left(\bf{y}-\sqrt{snr}\bf{H_1P_1x_1}-\sqrt{snr}\bf{H_2P_2x_2}\right)
\end{multline}
\begin{equation}
~~={p_{y|x_{1},x_{2}}}({\bf{y|x_{1},x_{2}}})\left({\bf{y}}-\sqrt{snr}{\bf{H_1P_1x_1}}-\sqrt{snr}{\bf{H_2P_2x_2}}\right)\sqrt{snr}{\bf{x_1^{\dag}P_1^{\dag}}}
\end{equation}
\begin{equation}
~~=-\nabla_{\bf{y}}{p_{y|x_{1},x_{2}}}({\bf{y|x_{1},x_{2}}})\sqrt{snr}{\bf{x_1^{\dag}P_1^{\dag}}}
\end{equation}

Therefore, we have:
\begin{equation}
\label{46}
\mathbb{E}_{x_1,x_2}\left[\nabla_{\bf{H_1^{\dag}}}{p_{y|x_{1},x_{2}}}({\bf{y|x_{1},x_{2}}})\right]=
\mathbb{E}_{x_1,x_2}\left[-\nabla_{\bf{y}}{p_{y|x_{1},x_{2}}}({\bf{y|x_{1},x_{2}}})\sqrt{snr}{\bf{x_1^{\dag}P_1^{\dag}}}\right]
\end{equation}

Substitute \eqref{46} into \eqref{45}, we get:
\begin{multline}
\frac{\partial I({\bf{x_1, x_2;y}})}{\partial {\bf{H_1^{\dag}}}}=\int \left(1+log\left({p_{y}}({\bf{y}})\right)\right)\mathbb{E}_{x_1,x_2}[\nabla_{\bf{y}}{p_{y|x_{1},x_{2}}}({\bf{y|x_{1},x_{2}}}) \times \\
\sqrt{snr}{\bf{x_1^{\dag}P_1^{\dag}}}] d\bf{y}
\end{multline}
\begin{equation}
~~=\mathbb{E}_{x_1,x_2}\left[\left(\int \left(1+log\left({p_{y}}({\bf{y}})\right)\right)\nabla_{\bf{y}}{p_{y|x_{1},x_{2}}}({\bf{y|x_{1},x_{2}}})d{\bf{y}}\right)\sqrt{snr}{\bf{x_1^{\dag}P_1^{\dag}}} \right]
\end{equation}

Using integration by parts applied to the real and imaginary parts of $\bf{y}$ we have:
\begin{multline}
\label{47}
\int \left(1+{log}\left({p_{y}}({\bf{y}})\right)\right) \frac{\partial {p_{y|x_{1},x_{2}}}({\bf{y|x_{1},x_{2}}})}{\partial t} d\bf{t}= \\
\int \left(1+{log}\left({p_{y}}({\bf{y}})\right)\right) {p_{y|x_{1},x_{2}}}({\bf{y|x_{1},x_{2}}}){|}_{-\infty}^{\infty}-\int_{-\infty}^{\infty} \frac{1}{{p_{y}}({\bf{y}})} \frac{\partial {p_{y}}({\bf{y}})}{\partial t}{p_{y|x_{1},x_{2}}}({\bf{y|x_{1},x_{2}}})d\bf{t}
\end{multline}

The first term in \eqref{47} goes to zero as $\left\|\bf{y}\right\| \rightarrow \infty$. Therefore,
\begin{equation}
\frac{\partial I({\bf{x_1, x_2;y}})}{\partial {\bf{H_1^{\dag}}}}=\mathbb{E}_{x_1,x_2}\left[-\int \left(\frac{{p_{y|x_{1},x_{2}}}({\bf{y|x_{1},x_{2}}})}{{p_{y}({\bf{y}})}}\nabla_{\bf{y}}{p_{y}({\bf{y}})}d{\bf{y}}\right)\sqrt{snr}{\bf{x_1^{\dag}P_1^{\dag}}}\right]
\end{equation}
\begin{equation}
\frac{\partial I({\bf{x_1, x_2;y}})}{\partial {\bf{H_1^{\dag}}}}=-\int \nabla_{\bf{y}}{p_{y}({\bf{y}})}\mathbb{E}_{x_1,x_2}\left[\frac{{p_{y|x_{1},x_{2}}}({\bf{y|x_{1},x_{2}}})}{{p_{y}({\bf{y}})}}\sqrt{snr}{\bf{x_1^{\dag}P_1^{\dag}}}\right]d\bf{y}
\end{equation}
\begin{equation}
\label{48}
\frac{\partial I({\bf{x_1, x_2;y}})}{\partial {\bf{H_1^{\dag}}}}=-\sqrt{snr} \int \nabla_{\bf{y}}{p_{y}({\bf{y}})}\mathbb{E}_{x_1|y}\left[{\bf{x_1|y}}\right]^{\dag}{\bf{P_1^{\dag}}}d\bf{y}
\end{equation}

However,
\begin{multline}
\label{49}
\nabla_{\bf{y}}{p_{y}({\bf{y}})}=\nabla_{\bf{y}}\mathbb{E}_{x_1,x_2}\left[{p_{y|x_{1},x_{2}}}({\bf{y|x_{1},x_{2}}})\right] \\
=\mathbb{E}_{x_1,x_2}\left[\nabla_{\bf{y}}{p_{y|x_{1},x_{2}}}({\bf{y|x_{1},x_{2}}})\right] \\
=-\mathbb{E}_{x_1,x_2}\left[p_{y|x_{1},x_{2}}({\bf{y|x_{1},x_{2}}})\left({\bf{y}}-\sqrt{snr}{\bf{H_1P_1x_1}}-\sqrt{snr}{\bf{H_2P_2x_2}}\right)\right]\\
=-\mathbb{E}_{x_1,x_2}\left[p_{y}({\bf{y}})\left({\bf{y}}-\sqrt{snr}{\bf{H_1P_1x_1}}-\sqrt{snr}{\bf{H_2P_2x_2}}\right)|{\bf{y}}\right] \\
=-p_{y}({\bf{y}})\left({\bf{y}}-\sqrt{snr}{\bf{H_1P_1}\mathbb{E}_{x_1|y}[{\bf{x_1|y}}]}-\sqrt{snr}\bf{H_2P_2}\mathbb{E}_{x_2|y}[{\bf{x_2|y}}]\right)\\
\end{multline}

Substitute \eqref{49} into \eqref{48} we get:
\begin{multline}
\frac{\partial I({\bf{x_1, x_2;y}})}{\partial {\bf{H_1^{\dag}}}}=\int {p_{y}}({\bf{y}})\left({\bf{y}}-\sqrt{snr}{\bf{H_1P_1}}\mathbb{E}_{x_1|y}[{\bf{x_1|y}}]-\sqrt{snr}{\bf{H_2P_2}}\mathbb{E}_{x_2|y}[{\bf{x_2|y}}]\right) \times \\
\sqrt{snr} \mathbb{E}_{x_1|y}\left[{\bf{x_1|y}}\right]^{\dag}{\bf{P_1^{\dag}}}d\bf{y}
\end{multline}
\begin{multline}
\frac{\partial I({\bf{x_1, x_2;y}})}{\partial {\bf{H_1^{\dag}}}}=\sqrt{snr}\mathbb{E}_y[{\bf{yx_1^{\dag}}}]{\bf{P_1^{\dag}}} \nonumber\\
-snr\mathbb{E}_y[{\bf{H_1P_1}}\mathbb{E}_{x_1|y}[{\bf{x_1|y}}]\mathbb{E}_{x_1|y}[{\bf{x_1|y}}]^{\dag}{\bf{P_1^{\dag}}} \nonumber\\
-snr\mathbb{E}_y[{\bf{H_2P_2}}\mathbb{E}_{x_2|y}[{\bf{x_2|y}}]\mathbb{E}_{x_2|y}[{\bf{x_1|y}}]^{\dag}]{\bf{P_1^{\dag}}}
\end{multline}

Therefore,
\begin{multline}
\frac{\partial I({\bf{x_1, x_2;y}})}{\partial {\bf{H_1^{\dag}}}}=\sqrt{snr}{\bf{H_1P_1}}\mathbb{E}_{x_1}[{\bf{x_1x_1^{\dag}}}]{\bf{P_1^{\dag}}} \nonumber\\
-snr{\bf{H_1P_1}}\mathbb{E}_y[{\mathbb{E}_{x_1|y}[{\bf{x_1|y}}]\mathbb{E}_{x_1|y}[{\bf{x_1|y}}]^{\dag}}{\bf{P_1^{\dag}}} \nonumber \\
-snr{\bf{H_2P_2}}\mathbb{E}_y[{\mathbb{E}_{x_2|y}[{\bf{x_2|y}}]\mathbb{E}_{x_1|y}[{\bf{x_1|y}}]^{\dag}}]{\bf{P_1^{\dag}}}
\end{multline}
\begin{equation}
\frac{\partial I({\bf{x_1, x_2;y}})}{\partial {\bf{H_1^{\dag}}}}=snr{\bf{H_1P_1}}{\bf{E_1}}{\bf{P_1^{\dag}}}-snr{\bf{H_2P_2}}\mathbb{E}_y[{\mathbb{E}_{x_2|y}[{\bf{x_2|y}}]\mathbb{E}_{x_1|y}[{\bf{x_1|y}}]^{\dag}}]{\bf{P_1^{\dag}}}
\end{equation}

Similarly, we can derive the gradient of the mutual information in terms of the channel matrix that corresponds to the second user in the two-user MAC. 
\begin{equation}
\frac{\partial I({\bf{x_1, x_2;y}})}{\partial {\bf{H_2^{\dag}}}}=snr{\bf{H_2P_2}}{\bf{E_2}}{\bf{P_2^{\dag}}}-snr{\bf{H_1P_1}}\mathbb{E}_y[{\mathbb{E}_{x_1|y}[{\bf{x_1|y}}]\mathbb{E}_{x_2|y}[{\bf{x_2|y}}]^{\dag}}]{\bf{P_2^{\dag}}}
\end{equation}

Therefore, the gradient of the mutual information with respect to per user channel and the per user MMSE and input estimates (or covariances) is as follows:
\begin{equation}
\nabla_{\bf{H_1}} I({\bf{x_1, x_2;y}})=snr{\bf{H_1P_1}}{\bf{E_1}}{\bf{P_1^{\dag}}}-snr{\bf{H_2P_2}}\mathbb{E}[{\bf{\widehat{x_2}{\widehat{x_1}}^{\dag}}}]{\bf{P_1^{\dag}}}
\end{equation}
\begin{equation}
\nabla_{\bf{H_2}} I({\bf{x_1, x_2;y}})=snr{\bf{{H_2P_2}}}{\bf{E_2}}{\bf{P_2^{\dag}}}-snr{\bf{H_1P_1}}\mathbb{E}[{\bf{\widehat{x_1}{\widehat{x_2}}^{\dag}}}]{\bf{P_2^{\dag}}}
\end{equation}

Therefore, Theorem~\ref{theorem1.4} has been proved.
\section{Appendix C: Proof of Theorem~\ref{theorem2.4}}
\normalsize
The gradient of the mutual information with respect to the precoding matrix of user 1 on the two-user MAC is as follows:
\begin{equation}
\frac{\partial I({\bf{x_1, x_2;y}})}{\partial {\bf{P_1^{\dag}}}}=-\frac{\partial}{\partial {\bf{P_1^{\dag}}}} \int {p_{y}}({\bf{y}})log\left(p_{y}({\bf{y}})\right) d\bf{y}
\end{equation}
\begin{equation}
~~=-\int \left({p_{y}}({\bf{y}})\frac{1}{{p_{y}}({\bf{y}})}+log\left(p_{y}({\bf{y}})\right)\right)\frac{\partial {p_{y}}({\bf{y}})}{\partial {\bf{P_1^{\dag}}}} d\bf{y}
\end{equation}
\begin{equation}
\label{45.2}
~~=-\int \left(1+\bf{log\left(p_{y}({\bf{y}})\right)}\right)\frac{\partial {p_{y}}({\bf{y}})}{\partial {\bf{P_1^{\dag}}}} d\bf{y}
\end{equation}
The derivative of the conditional output can be written as:
\begin{multline}
\frac{\partial {p_{y|x_{1},x_{2}}}({\bf{y|x_{1},x_{2}}})}{\partial {\bf{P_1^{\dag}}}}=\\
-{p_{y|x_{1},x_{2}}}({\bf{y|x_{1},x_{2}}})\frac{\partial }{\partial {\bf{P_1^{\dag}}}}\left({\bf{y}}-\sqrt{snr}{\bf{H_1P_1x_1}}-\sqrt{snr}{\bf{H_2P_2x_2}}\right)^{\dag} \times \\
\left({\bf{y}}-\sqrt{snr}{\bf{H_1P_1x_1}}-\sqrt{snr}{\bf{H_2P_2x_2}}\right)
\end{multline}
\begin{equation}
~~={p_{y|x_{1},x_{2}}}({\bf{y|x_{1},x_{2}}}){\bf{H_1^{\dag}}}\left({\bf{y}}-\sqrt{snr}{\bf{H_1P_1x_1}}-\sqrt{snr}\bf{H_2P_2x_2}\right)\bf{x_1^{\dag}}
\end{equation}
\begin{equation}
~~=-\sqrt{snr}\bf{H_1^{\dag}}\nabla_{\bf{y}}{p_{y|x_{1},x_{2}}}({\bf{y|x_{1},x_{2}}}){\bf{x_1^{\dag}}}
\end{equation}
Therefore, we have:
\begin{equation}
\label{46.2}
\mathbb{E}_{x_1,x_2}\left[\nabla_{\bf{P_1^{\dag}}}{p_{y|x_{1},x_{2}}}({\bf{y|x_{1},x_{2}}})\right]=
\mathbb{E}_{x_1,x_2}\left[-\sqrt{snr}{\bf{H_1^{\dag}}}\nabla_{\bf{y}}{p_{y|x_{1},x_{2}}}({\bf{y|x_{1},x_{2}}}){\bf{x_1^{\dag}}}\right]
\end{equation}

Substitute \eqref{46.2} into \eqref{45.2}, we get:
\begin{equation}
\frac{\partial I({\bf{x_1, x_2;y}})}{\partial {\bf{P_1^{\dag}}}}=\int \left(1+{log}\left({p_{y}}({\bf{y}})\right)\right)\mathbb{E}_{x_1,x_2}\left[\sqrt{snr}{\bf{H_1^{\dag}}}\nabla_{\bf{y}}{p_{y|x_{1},x_{2}}}({\bf{y|x_{1},x_{2}}}){\bf{x_1^{\dag}}}\right] d\bf{y}
\end{equation}
\begin{equation}
~~=\mathbb{E}_{x_1,x_2}\left[\sqrt{snr}{\bf{H_1^{\dag}}} \left(\int \left(1+{log}\left({p_{y}}({\bf{y}})\right)\right)\nabla_{\bf{y}}{p_{y|x_{1},x_{2}}}({\bf{y|x_{1},x_{2}}})d\bf{y}\right){\bf{x_1^{\dag}}} \right]
\end{equation}

Repeating the same steps as in \eqref{47}, we have:
\begin{equation}
\frac{\partial I({\bf{x_1, x_2;y}})}{\partial {\bf{P_1^{\dag}}}}=\mathbb{E}_{x_1,x_2}\left[-\int \sqrt{snr}{\bf{H_1^{\dag}}}\left(\frac{{p_{y|x_{1},x_{2}}}({\bf{y|x_{1},x_{2}}})}{{p_{y}({\bf{y}})}}\nabla_{\bf{y}}{p_{y}({\bf{y}})}d\bf{y}\right){\bf{x_1^{\dag}}}\right]
\end{equation}
\begin{equation}
\frac{\partial I({\bf{x_1, x_2;y}})}{\partial {\bf{P_1^{\dag}}}}=-\int \sqrt{snr} {\bf{H_1^{\dag}}} \nabla_{\bf{y}}{p_{y}({\bf{y}})}\mathbb{E}_{x_1,x_2}\left[\frac{{p_{y|x_{1},x_{2}}}({\bf{y|x_{1},x_{2}}})}{{p_{y}({\bf{y}})}}{\bf{x_1^{\dag}}}\right]d\bf{y}
\end{equation}
\begin{equation}
\label{48.2}
\frac{\partial I({\bf{x_1, x_2;y}})}{\partial {\bf{P_1^{\dag}}}}=-\sqrt{snr}\int {\bf{H_1^{\dag}}}\nabla_{\bf{y}}{p_{y}({\bf{y}})}\mathbb{E}_{x_1|y}\left[{\bf{x_1|y}}\right]^{\dag} d\bf{y}
\end{equation}

Substitute \eqref{49} into \eqref{48.2} we get:
\begin{multline}
\frac{\partial I({\bf{x_1, x_2;y}})}{\partial {\bf{P_1^{\dag}}}}=\int \sqrt{snr}{\bf{H_1^{\dag}}}{p_{y}}({\bf{y}}) \times \nonumber \\
\left({\bf{y}}-\sqrt{snr}{\bf{H_1P_1}}\mathbb{E}_{x_1|y}[{\bf{x_1|y}}]-\sqrt{snr}{\bf{H_2P_2}}\mathbb{E}_{x_2|y}[{\bf{x_2|y}}]\right) \mathbb{E}_{x_1|y}\left[{\bf{x_1|y}}\right]^{\dag}d\bf{y}
\end{multline}
\begin{multline}
\frac{\partial I({\bf{x_1, x_2;y}})}{\partial {\bf{P_1^{\dag}}}}=\sqrt{snr}{\bf{H_1^{\dag}}}\mathbb{E}_y[{\bf{yx_1^{\dag}}}]-snr{\bf{H_1^{\dag}}}\mathbb{E}_y[{\bf{H_1P_1}}\mathbb{E}_{x_1|y}[{\bf{x_1|y}}]\mathbb{E}_{x_1|y}[{\bf{x_1|y}}]^{\dag} \nonumber \\
-snr{\bf{H_1^{\dag}}}\mathbb{E}_y[{\bf{H_2P_2}}\mathbb{E}_{x_2|y}[{\bf{x_2|y}}]\mathbb{E}_{x_1|y}[{\bf{x_1|y}}]^{\dag}]
\end{multline}

Therefore,
\begin{align}
\frac{\partial I({\bf{x_1, x_2;y}})}{\partial {\bf{P_1^{\dag}}}}=snr{\bf{H_1^{\dag}}}{\bf{H_1P_1}}\mathbb{E}_{x_1}[{\bf{x_1x_1^{\dag}}}]-snr{\bf{H_1^{\dag}}}{\bf{H_1P_1}}\mathbb{E}_y[{\mathbb{E}_{x_1|y}[{\bf{x_1|y}}]\mathbb{E}_{x_1|y}[{\bf{x_1|y}}]^{\dag}} \nonumber\\
-snr{\bf{H_1^{\dag}}}{\bf{H_2P_2}}\mathbb{E}_y[{\mathbb{E}_{x_2|y}[{\bf{x_2|y}}]\mathbb{E}_{x_1|y}[{\bf{x_1|y}}]^{\dag}}]
\end{align}
\begin{equation}
\frac{\partial I({\bf{x_1, x_2;y}})}{\partial {\bf{P_1^{\dag}}}}=snr{\bf{H_1^{\dag}}}{\bf{H_1P_1}}{\bf{E_1}}-snr{\bf{H_1^{\dag}}}{\bf{H_2P_2}}\mathbb{E}_y[{\mathbb{E}_{x_2|y}[{\bf{x_2|y}}]\mathbb{E}_{x_1|y}[{\bf{x_1|y}}]^{\dag}}]
\end{equation}

Similarly, we can derive the gradient of the mutual information in terms of the channel matrix that corresponds to the second user in the two-user MAC. 
\begin{equation}
\frac{\partial I({\bf{x_1, x_2;y}})}{\partial {\bf{P_2^{\dag}}}}=snr{\bf{H_2^{\dag}}}{\bf{H_2P_2}}{\bf{E_2}}-snr{\bf{H_2^{\dag}}}{\bf{H_1P_1}}\mathbb{E}_y[{\mathbb{E}_{x_1|y}[{\bf{x_1|y}}]\mathbb{E}_{x_2|y}[{\bf{x_2|y}}]^{\dag}]}
\end{equation}

Therefore, the gradient of the mutual information with respect to per user precoding and the per user MMSE and input estimates (or covariances) is as follows:
\begin{equation}
\label{5.1}
\nabla_{\bf{P_1}}I({\bf{x_{1},x_{2};y}})=snr{\bf{H_1^{\dag}}}{\bf{H}_{1}}{\bf{P}_1}{\bf{E}_{1}}-snr{\bf{H_1^{\dag}}}{\bf{H}_{2}}{\bf{P}_2} \mathbb{E}_y[{\bf{\widehat{x}_2\widehat{x}_1^{\dag}}}]
\end{equation}
\begin{equation}
\label{6.1}
\nabla_{\bf{P_2}}I({\bf{x_{1},x_{2};y}})=snr{\bf{H_2^{\dag}}}{\bf{H}_{2}}{\bf{P}_2}{\bf{E}_{2}}-snr{\bf{H_2^{\dag}}}{\bf{H}_{1}}{\bf{P}_1}\mathbb{E}_y[{\bf{\widehat{x}_1\widehat{x}_2^{\dag}}}]
\end{equation}

Therefore, Theorem~\ref{theorem2.4} has been proved.
\section{Appendix D: Proof of Theorem~\ref{theorem3.4}}
\normalsize
From the steps in Theorem~\ref{theorem1.4}, we can see that:
\begin{multline}
{p_{y}}({\bf{y}})\left({\bf{y}}-\sqrt{snr}{\bf{H_1P_1}}\mathbb{E}_{x_1|y}[{\bf{x_1|y}}]-\sqrt{snr}{\bf{H_2P_2}}\mathbb{E}_{x_2|y}[{\bf{x_2|y}}]\right) \\
=\mathbb{E}_{x_1,x_2}\left[{p_{y|x_{1},x_{2}}}({\bf{y|x_{1},x_{2}}})\left({\bf{y}}-\sqrt{snr}{\bf{H_1P_1x_1}}-\sqrt{snr}{\bf{H_2P_2x_2}}\right)\right]
\end{multline}

Therefore,
\begin{multline}
\sqrt{snr}{\bf{H_1P_1}}\mathbb{E}_{x_1|y}[{\bf{x_1|y}}]-\sqrt{snr}{\bf{H_2P_2}}\mathbb{E}_{x_2|y}[{\bf{x_2|y}}] \\
={\bf{y}}- \frac{\mathbb{E}_{x_1,x_2}\left[{p_{y|x_{1},x_{2}}}({\bf{y|x_{1},x_{2}}})\left({\bf{y}}-\sqrt{snr}{\bf{H_1P_1x_1}}-\sqrt{snr}{\bf{H_2P_2x_2}}\right)\right]}{{p_{y}}({\bf{y}})}  \\
={\bf{y}}+\frac{\mathbb{E}_{x_1,x_2}\left[\nabla_{{\bf{y}}}{p_{y|x_{1},x_{2}}}({\bf{y|x_{1},x_{2}}})\right]}{{p_{y}}({\bf{y}})} \\
={\bf{y}}+\frac{\nabla_{{\bf{y}}}\mathbb{E}_{x_1,x_2}\left[{p_{y|x_{1},x_{2}}}({\bf{y|x_{1},x_{2}}})\right]}{{p_{y}}({\bf{y}})}
\end{multline}

Thus,
\begin{equation}
\sqrt{snr}{\bf{H_1P_1}}\mathbb{E}[{\bf{x_1|y}}]-\sqrt{snr}{{\bf{H_2P_2}}}\mathbb{E}[{\bf{x_2|y}}]={\bf{y}}+\frac{\nabla_{{\bf{y}}}{p_{y}}({\bf{y}})}{{p_{y}}({\bf{y}})}
\end{equation}

Therefore, Theorem~\ref{theorem3.4} has been proved. 

\section{Appendix E: Proof of Theorem~\ref{theorem4.00}}
The MMSE Wiener filters are known to be MSE minimizers. In fact, this kind of receive filters inherently include the linear MMSE matrix. Therefore, to start the proof of the theorem, lets first define the linear MMSE as:
\begin{equation}
mmse(snr)=\underbrace{\bf{\mathbb{E}}_{y}\left[\left\|\bf{H_1P_1}(\bf{x_1-\widehat{x}_1})\right\|^{2}\right]}_\text{MMSE1}+\underbrace{\bf{\mathbb{E}}_{y}\left[\left\|\bf{H_2P_2}(\bf{x_2-\widehat{x}_2})\right\|^{2}\right]}_\text{MMSE2},
\end{equation}

with the linear estimates should be found by multiplying the received vector $\bf{y}$ by the receive filter $\bf{H_r}$,
\begin{equation}
\bf{\widehat{x}_1}=H_{r1}y
\end{equation}
\begin{equation}
\bf{\widehat{x}_2}=H_{r2}y
\end{equation}

We break down the problem into two parts for each per user MMSE. Notice that we call it MMSE since our target is the minimum MSE, therefore, we first take the first part, i.e., MSE1 and break it down; it follows that:
\begin{equation}
\bf{\tilde{E_1}}=\bf{\mathbb{E}}_{y}\left[\left\|\bf{H_1P_1}(\bf{x_1-\widehat{x}_1})\right\|^{2}\right]
\end{equation}
\begin{equation}
~~=\bf{\mathbb{E}}_{y}\left[\bf{H_1P_1}(\bf{x_1-\widehat{x}_1})(\bf{x_1-\widehat{x}_1})^{\dag}(\bf{H_1P_1})^{\dag}\right]
\end{equation}
\begin{equation}
~~=\bf{\mathbb{E}}_{y}\left[(\bf{\bf{H_1P_1}x_1-\bf{H_1P_1}\widehat{x}_1})(\bf{\bf{H_1P_1}x_1-\bf{H_1P_1}\widehat{x}_1})^{\dag}\right]
\end{equation}
\begin{multline}
\label{4.4.1}
~~=\underbrace{\bf{\mathbb{E}}_{x_1}\left[\bf{H_1P_1x_1{x_1}^{\dag}}(\bf{H_1P_1})^{\dag}\right]}_\text{first term}
-\underbrace{\bf{\mathbb{E}}_{x_1}\left[\bf{H_1P_1x_1\widehat{x}_1^{\dag}}(\bf{H_1P_1})^{\dag}\right]}_\text{second term}  \\
-\underbrace{\bf{\mathbb{E}}_{x_1}\left[\bf{H_1P_1\widehat{x}_1x_1^{\dag}}(\bf{H_1P_1})^{\dag}\right]}_\text{third term}
+\underbrace{\bf{\mathbb{E}}_{x_1}\left[\bf{H_1P_1\widehat{x}_1\widehat{x}_1^{\dag}}(\bf{H_1P_1})^{\dag}\right]}_\text{forth term}\\
\end{multline}

Lets work on each term in \eqref{4.4.1} separately. Notice that $\mathbb{E}_{x_1}[\bf{x_1x_1^{\dag}}]=\mathbb{E}_{x_2}[\bf{x_2x_2^{\dag}}]=I$, $\mathbb{E}_{x_1}[\bf{x_1x_1^{T}}]=\mathbb{E}_{x_2}[\bf{x_2x_2^{T}}]=0$, $\mathbb{E}_{x_1,x_2}[\bf{x_2x_1^{\dag}}]=\mathbb{E}_{x_1,x_2}[\bf{x_1x_2^{\dag}}]=0$, $\mathbb{E}_n[\bf{nn^{\dag}}]=I$, and $\mathbb{E}_n[\bf{nn^{T}}]=0$.

The first term is:
\begin{equation}
\bf{\mathbb{E}}_{x_1}\left[(\bf{H_1P_1x_1{x_1}^{\dag}}(\bf{H_1P_1})^{\dag}\right]=\bf{\mathbb{E}}_{x_1}\left[\bf{H_1P_1}(\bf{H_1P_1})^{\dag}\right],
\end{equation}

The second term is:
\begin{equation}
\bf{\mathbb{E}}_{x_1}\left[\bf{H_1P_1x_1\widehat{x}_1^{\dag}}(\bf{H_1P_1})^{\dag}\right]=\bf{\mathbb{E}}_{x_1}\left[\bf{H_1P_1x_1}(\bf{H_{r1}y})^{\dag}(\bf{H_1P_1})^{\dag}\right]
\end{equation}
\begin{equation}
\bf{\mathbb{E}}_{x_1}\left[\bf{H_1P_1x_1}\bf{{y}^{\dag}H_{r1}^{\dag}}(\bf{H_1P_1})^{\dag}\right]=\bf{\mathbb{E}}_{x_1}\left[\bf{H_1P_1x_1}\bf{(\bf{H_1P_1x_1+H_2P_2x_2+n})^{\dag}H_{r1}^{\dag}}(\bf{H_1P_1})^{\dag}\right]
\end{equation}
\begin{equation}
~~=\bf{\mathbb{E}}_{x_1}\left[\bf{H_1P_1}\bf{P_1^{\dag}H_1^{\dag}H_{r1}^{\dag}}(\bf{H_1P_1})^{\dag}\right]
\end{equation}

The third term is:
\begin{equation}
\bf{\mathbb{E}}_{x_1}\left[\bf{H_1P_1\widehat{x}_1x_1^{\dag}}(\bf{H_1P_1})^{\dag}\right]=\bf{\mathbb{E}}_{x_1}\left[\bf{H_1P_1H_{r1}yx_1^{\dag}}(\bf{H_1P_1})^{\dag}\right]
\end{equation}
\begin{equation}
~~=\bf{\mathbb{E}}_{x_1}\left[\bf{H_1P_1H_{r1}H_1P_1}(\bf{H_1P_1})^{\dag}\right]
\end{equation}

The forth term is:
\begin{equation}
\bf{\mathbb{E}}_{x_1}\left[\bf{H_1P_1\widehat{x}_1\widehat{x}_1^{\dag}}(\bf{H_1P_1})^{\dag}\right]=\bf{\mathbb{E}}_y\left[\bf{H_1P_1H_{r1}yy^{\dag}H_{r1}^{\dag}}(\bf{H_1P_1})^{\dag}\right]
\end{equation}
\begin{multline}
\bf{\mathbb{E}}_y\left[\bf{H_1P_1H_{r1}yy^{\dag}H_{r1}^{\dag}}(\bf{H_1P_1})^{\dag}\right]=\bf{\mathbb{E}}_y\left[\bf{H_1P_1H_{r1}H_1P_1P_1^{\dag}H_1^{\dag}H_{r1}^{\dag}}(\bf{H_1P_1})^{\dag}\right] \\
+\bf{\mathbb{E}}_y\left[\bf{H_1P_1H_{r1}H_2P_2P_2^{\dag}H_2^{\dag}H_{r1}^{\dag}}(\bf{H_1P_1})^{\dag}\right]
+\bf{\mathbb{E}}_y\left[\bf{H_1P_1H_{r1}H_{r1}^{\dag}}(\bf{H_1P_1})^{\dag}\right]
\end{multline}

Therefore, the MSE1 can be given as follows:
\begin{multline}
MSE1(snr)=Tr\left\{\bf{H_1P_1}(\bf{H_1P_1})^{\dag}\right\} -Tr\left\{\bf{H_1P_1}\bf{P_1^{\dag}H_1^{\dag}H_{r1}^{\dag}}(\bf{H_1P_1})^{\dag}\right\} \\
-Tr\left\{\bf{H_1P_1H_{r1}H_1P_1}(\bf{H_1P_1})^{\dag}\right\} 
+Tr\left\{\bf{H_1P_1H_{r1}H_1P_1P_1^{\dag}H_1^{\dag}H_{r1}^{\dag}}(\bf{H_1P_1})^{\dag}\right\} \\
+Tr\left\{\bf{H_1P_1H_{r1}H_2P_2P_2^{\dag}H_2^{\dag}H_{r1}^{\dag}}(\bf{H_1P_1})^{\dag}\right\}
+Tr\left\{\bf{H_1P_1H_{r1}H_{r1}^{\dag}}(\bf{H_1P_1})^{\dag}\right\}
\end{multline}

We need now to find the minimum MSE1, and apply the KKT conditions to derive the optimal MMSE1 Wiener receive filter $\bf{H_{r1}}$ as follows:
\begin{equation}
\frac{\partial MSE1}{\partial H_{r1}}=0
\end{equation}

Capitalizing on the derivative rules for the trace of matrices, $\frac{\partial Tr(ABA^{\dag})}{\partial A}= AB^{\dag}$, $\frac{\partial Tr(AB)}{\partial A}= B^{\dag}$, and $\frac{\partial Tr(BA^{\dag})}{\partial A}= B$. We have:
\begin{equation}
\bf{H_{r1}^{\star}}=\bf{P_1^{\dag}H_1^{\dag}(I+P_1^{\dag}H_1^{\dag}P_1H_1+P_2^{\dag}H_2^{\dag}P_2H_2)^{-1}}
\end{equation}

Similar steps to break down MSE2,  
\begin{equation}
\bf{\tilde{E_2}}=\bf{\mathbb{E}}_y\left[\left\|\bf{H_2P_2}(\bf{x_2-\widehat{x}_2})\right\|^{2}\right],
\end{equation}

leads to the optimal MMSE2 Wiener filter $\bf{H_{r2}}$ as follows:
\begin{equation}
\bf{H_{r2}^{\star}}=\bf{P_2^{\dag}H_2^{\dag}(I+P_1^{\dag}H_1^{\dag}P_1H_1+P_2^{\dag}H_2^{\dag}P_2H_2)^{-1}} 
\end{equation}

Therefore, to estimate the user inputs linearly we do the following:
\begin{equation}
\bf{\widehat{x}_1}=\bf{P_1^{\dag}H_1^{\dag}(I+P_1^{\dag}H_1^{\dag}P_1H_1+P_2^{\dag}H_2^{\dag}P_2H_2)^{-1}y}
\end{equation}
\begin{equation}
\bf{\widehat{x}_2}=\bf{P_2^{\dag}H_2^{\dag}(I+P_1^{\dag}H_1^{\dag}P_1H_1+P_2^{\dag}H_2^{\dag}P_2H_2)^{-1}y}
\end{equation}

Therefore, Theorem~\ref{theorem4.00} has been proved.

\section{Appendix F: Proof of Theorem~\ref{theorem5.4}}
The possible solutions to ~\eqref{2} subject to~\eqref{3} and~\eqref{4} are characterized by the KKT conditions, which give necessary conditions for the matrix to be a critical point, known as the KKT or first-order conditions. To compute the KKT conditions, we first build the Lagrangian:
\vspace{-0.1cm}
\begin{equation}
\mathcal{L}({\bf{P_{1}}},{\bf{P_{2}}},\lambda_1,\lambda_2)=-I({\bf{x_{1},x_{2};y}})-\lambda_1(Q_1-{\bf{P_1}})-\lambda_2(Q_2-{\bf{P_2}})-\mu_1 {\bf{P_1}}-\mu_2 {\bf{P_2}} 
\end{equation}

With primal feasibility conditions, $\lambda_1(Q_1-{\bf{P_{1}}})=0$, $\mu_1{\bf{P_{1}}}=0$, $\lambda_2(Q_2-{\bf{P_{2}}})=0$, and $\mu_2{\bf{P_{2}}}=0$, and dual feasibility conditions, $\lambda_1\geq0$ and $\lambda_2\geq0$. In which the Lagrange multipliers $\lambda_1$ and $\lambda_2$ accounting for the in-equality constraint, has to be non-negative. The first-order conditions are given by:
\begin{equation}
\nabla_{\bf{P_{1}}}\mathcal{L}({\bf{P_{1}}},{\bf{P_{2}}},\lambda_1,\lambda_2)=-\nabla_{\bf{P_{1}}}I({\bf{x_{1},x_{2};y}})+\lambda_1 {\bf{P_1}}=0
\end{equation}
\begin{equation}
\nabla_{\bf{P_{2}}}\mathcal{L}({\bf{P_{1}}},{\bf{P_{2}}},\lambda_1,\lambda_2)=-\nabla_{\bf{P_{2}}}I({\bf{x_{1},x_{2};y}})+\lambda_2 {\bf{P_2}}=0
\end{equation}

Therefore, the solution of the optimal precoders satisfies:
\begin{equation}
\bf{{P_1}^{\star}}=\nu_1^{-1}{\bf{H}_{1}}^{\dag}{\bf{H}}_{1}{\bf{P}_1}^{\star}{\bf{E}_{1}}-\nu_1^{-1}\bf{H_{1}}^{\dag}{\bf{H}_{2}}{\bf{P}_2}^{\star}\bf{\mathbb{E}[\widehat{x}_2\widehat{x}_1^{\dag}]}
\end{equation}
\begin{equation}
\bf{{P_2}^{\star}}=\nu_2^{-1}{\bf{H}_{2}}^{\dag}{\bf{H}}_{2}{\bf{P}_2}^{\star}{\bf{E}_{2}}-\nu_2^{-1}\bf{H_{2}}^{\dag}{\bf{H}_{1}}{\bf{P}_1}^{\star}\bf{\mathbb{E}[\widehat{x}_1\widehat{x}_2^{\dag}]}
\end{equation}

where the solution follows from the gradient of the mutual information with respect to the precoding matrices. We eliminate $\lambda_1$ and $\lambda_2$ by introducing $\nu_1=\frac{\lambda_1}{snr}$, and $\nu_2=\frac{\lambda_2}{snr}$. Therefore, Theorem~\ref{theorem5.4} has been proved.

\section{Appendix G: Proof of Theorem~\ref{theorem6.4}}
Digging into the depth of equations~\eqref{11} and \eqref{12}, we can do a singular value decomposition of the channel matrix $\bf{H_1}=\bf{U_{H_1}}\bf{\Lambda_{H_1}}\bf{V_{H_1}^{\dag}}$, and $\bf{H_2}=\bf{U_{H_2}}\bf{\Lambda_{H_2}}\bf{V_{H_2}^{\dag}}$, and the eigen value decomposition of the MMSE matrices, $\bf{E_1}=\bf{U_{E_1}}\bf{\Lambda_{E_1}}\bf{V_{E_1}^{\dag}}$, and $\bf{E_2}=\bf{U_{E_2}}\bf{\Lambda_{E_2}}\bf{V_{E_2}^{\dag}}$, we assume that the second terms in equations \eqref{9} and \eqref{10} are zeros. Therefore, multiplying the first term of \eqref{9} by $\bf{P_1^{\star\dag}}$, and of \eqref{10} by $\bf{P_2^{\star\dag}}$, and taking the trace, which corresponds to each of the users power constraints. It follows that: 
\begin{align}
Tr\left\{{{\bf{P}}_1}^{\star}{{\bf{P}}_1}^{\star\dag}\right\}
&=Tr\left\{\bf{H_{1}}\bf{H_{1}^{\dag}}\bf{P_{1}^{\star}}\bf{E_{1}}{{\bf{P}}_1}^{\star\dag}\right\}  \nonumber\\
&=Tr\left\{{{\bf{P}}_1}^{\star\dag}\bf{H_{1}}\bf{H_{1}^{\dag}}\bf{P_{1}^{\star}}\bf{E_{1}}\right\}  \nonumber\\
&=Tr\left\{\bf{E_{1}}{{\bf{P}}_1}^{\star\dag}\bf{H_{1}}\bf{H_{1}^{\dag}}\bf{P_{1}^{\star}}\right\}  \nonumber\\
&=Tr\left\{\bf{P_{1}^{\star}}\bf{E_{1}}{{\bf{P}}_1}^{\star\dag}\bf{H_{1}}\bf{H_{1}^{\dag}}\right\}  \nonumber\\
&=Tr\left\{\bf{\Lambda_{P_1}}\bf{\Lambda_{E_1}}\bf{\Lambda_{P_1}}\bf{\Lambda_{H_1}^{2}}\right\}  \nonumber\\
&=Tr\left\{\bf{\Lambda_{H_1}^{2}}\bf{\Lambda_{P_1}}\bf{\Lambda_{E_1}}\bf{\Lambda_{P_1}}\right\}  \nonumber\\
&=Tr\left\{\bf{\Lambda_{H_1}^{2}}\bf{\Lambda_{P_1}^{2}}\bf{\Pi\Lambda_{E_1}}\right\} \nonumber\\ 
\end{align}

With:
\begin{equation}
\bf{U_1=V_{H_1}},
\end{equation}
\begin{equation}
\bf{U_2=V_{H_2}},
\end{equation}
\begin{equation}
\bf{D_1}=diag({\sqrt{p_{1,1}}},~...~,{\sqrt{p_{1,n_t}}}),
\end{equation}
\begin{equation}
\bf{D_2}=diag({\sqrt{p_{2,1}}},~...~,{\sqrt{p_{2,n_t}}}),
\end{equation}
\begin{equation}
\bf{R_1}=\bf{\Pi U_{E_1}},
\end{equation}
\begin{equation}
\bf{R_2}=\bf{\Pi U_{E_2}}.
\end{equation}

Similar steps to prove the setup of the fixed point equation of the second user. Therefore, the fixed point equation of the per user optimal precoder precluding the other users interference is given as: 
\begin{equation}
\bf{P_1}=\bf{U_1}\bf{D_1}\bf{R_{1}^{\dag}}
\end{equation}
\begin{equation}
\bf{P_2}=\bf{U_2}\bf{D_2}\bf{R_{2}^{\dag}}
\end{equation}

Therefore, Theorem~\ref{theorem6.4} is proved. Note that we provide another proof for this theorem in \cite{107}.

\section{Appendix H: Proof of Theorem~\ref{theorem7.4}}
The Lagrangian for the optimization problem \eqref{31} subject to \eqref{32} to \eqref{34} is given by:
\begin{equation}
\mathcal{L}({\bf{P_1}},{\bf{P_2}},\lambda_1,\lambda_2)=-I({\bf{x_{1},x_{2};y}})-\lambda_1({1-\sum\limits_{j=1}^{n_t} {p_{1j}}})-\lambda_2({1-\sum\limits_{j=1}^{n_t} {p_{2j}}})- \sum\limits_{j=1}^{n_t} \mu_1{p_{1j}} -\sum\limits_{j=1}^{n_t} \mu_2{p_{2j}} 
\end{equation}

With primal feasibility conditions, $\lambda_1(1-\sum\limits_{j=1}^{n_t} {p_{1j}})=0$, $\mu_1\sum\limits_{j=1}^{n_t} {p_{1j}}=0$, $\lambda_2(1-\sum\limits_{j=1}^{n_t} {p_{2j}} )=0$, and $\mu_2\sum\limits_{j=1}^{n_t} {p_{2j}} =0$, and dual feasibility conditions, $\lambda_1\geq0$ and $\lambda_2\geq0$. The KKT conditions are given by:
\begin{equation}
\frac{\partial\mathcal{L}({\bf{P_1}},{\bf{P_2}},\lambda_1,\lambda_2)}{\partial \bf{p_1}}=-\frac{\partial I(x_{1},x_{2};y)}{\partial \bf{p_1}}+\lambda_1-\mu_1=0
\end{equation}
\begin{equation}
\frac{\partial \mathcal{L}({\bf{P_1}},{\bf{P_2}},\lambda_1,\lambda_2)}{\partial \bf{p_2}}=-\frac{\partial I({\bf{x_{1},x_{2};y}})}{\partial \bf{p_2}}+\lambda_2-\mu_2=0
\end{equation}

Note that $\mu_1=0$ for $\bf{p_1} \succ 0$ and $\mu_2=0$ for $\bf{p_2}\succ 0$. Therefore, the derivative with respect to $\bf{p_{1}}$, and $\bf{p_{2}}$ respectively, is given by:
\begin{equation}
\frac{\partial\mathcal{L}({\bf{P_1}},{\bf{P_2}},\lambda_1,\lambda_2)}{\partial \bf{p_1}}=\frac{snr}{\sqrt{\bf{p_1}}}\left({\bf{P_1}^{\star}}{\bf{H}}_{1}^{\dag}{\bf{H}_{1}}{\bf{P}_1}^{\star}{\bf{E}_{1}}-{\bf{P_1}^{\star}}{\bf{H}}_{1}^{\dag}{\bf{H}_{2}}{\bf{P}_2}^{\star}\bf{\mathbb{E}[\widehat{x_2}\widehat{x_1}^{\dag}]}\right)_{1j}
\end{equation}
\begin{equation}
\frac{\partial\mathcal{L}({\bf{P_1}},{\bf{P_2}},\lambda_1,\lambda_2)}{\partial \bf{p_2}}=\frac{snr}{\sqrt{\bf{p_2}}}\left({\bf{P_2}^{\star}}{\bf{H}}_{2}^{\dag}{\bf{H}_{2}}{\bf{P}_2}^{\star}{\bf{E}_{2}}-{\bf{P_2}^{\star}}{\bf{H}}_{2}^{\dag}{\bf{H}_{1}}{\bf{P}_1}^{\star}\bf{\mathbb{E}[\widehat{x_1}\widehat{x_2}^{\dag}]}\right)_{2j}
\end{equation}

Therefore, the optimal power allocation satisfies:
\begin{equation}
{{\bf{p}}_{1}}^{\star}=\gamma_1^{-1}\left({\bf{P_1}^{\star}}{\bf{H}}_{1}^{\dag}{\bf{H}_{1}}{\bf{P}_1}^{\star}{\bf{E}_{1}}-{\bf{P_1}^{\star}}{\bf{H}}_{1}^{\dag}{\bf{H}_{2}}{\bf{P}_2}^{\star}\bf{\mathbb{E}[\widehat{x_2}\widehat{x_1}^{\dag}]}\right)_{1j}
\end{equation}
\begin{equation}
{{\bf{p}}_{2}}^{\star}=\gamma_2^{-1}\left({\bf{P_2}^{\star}}{\bf{H}}_{2}^{\dag}{\bf{H}_{2}}{\bf{P}_2}^{\star}{\bf{E}_{2}}-{\bf{P_2}^{\star}}{\bf{H}}_{2}^{\dag}{\bf{H}_{1}}{\bf{P}_1}^{\star}\bf{\mathbb{E}[\widehat{x_1}\widehat{x_2}^{\dag}]}\right)_{2j}
\end{equation}

where $\gamma_1=\frac{\lambda_1}{snr}$, and $\gamma_2=\frac{\lambda_2}{snr}$. Therefore, Theorem~\ref{theorem7.4} has been proved.

\bibliographystyle{IEEEtran}
\bibliography{IEEEabrv,mybibfile}

\end{document}